\documentclass[9pt]{article}
\usepackage[utf8]{inputenc}
\usepackage[a4paper,margin=1in]{geometry}

\usepackage{makeidx}  
\usepackage{amsmath}
\usepackage{amssymb}
\usepackage{mathtools}
\usepackage{multirow}
\usepackage{graphicx}
\usepackage{algorithmic}
\usepackage{algorithm}
\usepackage{array,booktabs}
\usepackage{multirow}
\usepackage[none]{hyphenat}
\usepackage[english]{babel}
\usepackage{url}

\usepackage{amsthm}
\theoremstyle{definition}
\newtheorem{theorem}{Theorem}[section]
\newtheorem{definition}{Definition}[section]
\newtheorem{lemma}[theorem]{Lemma}

\emergencystretch=\maxdimen
\title{A parametric level-set method for partially discrete tomography}
\author{Ajinkya Kadu\footnotemark[2]~, Tristan van Leeuwen\footnotemark[2]~, K. Joost Batenburg\footnotemark[3]}
\date{}

\DeclarePairedDelimiter{\diagpars}{(}{)}
\newcommand{\diag}{\operatorname{diag}\diagpars}
\newcommand{\argmin}{\operatornamewithlimits{arg min}}

\begin{document}

\maketitle              

\let\oldthefootnote\thefootnote
\renewcommand{\thefootnote}{\fnsymbol{footnote}}
\footnotetext[1]{Paper submitted to $20^{th}$ International Conference 
on \textit{Discrete Geometry for Computer Imagery}}
\footnotetext[2]{Mathematical Institute, Utrecht University, The Netherlands. Contact: \url{ajinkyakadu125@gmail.com}}
\footnotetext[3]{Centrum Wiskunde \& Informatica, Amsterdam, The Netherlands}

\let\thefootnote\oldthefootnote

\begin{abstract}
 This paper introduces a parametric level-set method for tomographic reconstruction of partially discrete images. Such images consist of a continuously varying background and an anomaly with a constant (known) grey-value. We represent the geometry of the anomaly using a level-set function, which we represent using radial basis functions. We pose the reconstruction problem as a bi-level optimization problem in terms of the background and coefficients for the level-set function. To constrain the background reconstruction we impose smoothness through Tikhonov regularization. The bi-level optimization problem is solved in an alternating fashion; in each iteration we first reconstruct the background and consequently update the level-set function. We test our method on numerical phantoms and show that we can successfully reconstruct the geometry of the anomaly, even from limited data. On these phantoms, our method outperforms Total Variation reconstruction, DART and P-DART.

\end{abstract}

\section{Introduction}
The need to reconstruct (quantitative) images of an object from tomographic measurements appears in many applications. At the heart of many of these applications is a projection model based on the Radon transform. Characterizing the object under investigation by a function $u(\mathbf{x})$ with $\mathbf{x}\in\mathcal{D}=[0,1]^2$, tomographic measurements are modeled as
\[
p_i = \int_{\mathcal{D}}\, u(\mathbf{x})\delta(s_i - \mathbf{n}(\theta_i)\cdot\mathbf{x})\,\mathrm{d}\mathbf{x},
\]
where $s_i \in [0,1]$ denotes the shift, $\theta_i\in[0,2\pi)$ denotes the angle and $\mathbf{n}(\theta) = (\cos\theta,\sin\theta)$. The goal is to retrieve $u$ from a number, $m$, of such measurements for various shifts and directions.

If the shifts and angles are regularly sampled, the transform can be inverted directly by Filtered back-projection or Fourier reconstruction \cite{kak2001principles}. A common approach for dealing with non-regularly sampled or missing data, is to express $u$ in terms of a basis
\[
u(\mathbf{x}) = \sum_{j=1}^n u_jb(\mathbf{x} - \mathbf{x}_j),
\]
where $b$ are piece-wise polynomial basis functions and $\{\mathbf{x}_j\}_{j=1}^n$ is a regular (pixel) grid. This leads to a set of $m$ linear equations in $n$ unknowns
\[
\mathbf{p} = W\mathbf{u},
\]
with $w_{ij} = \int_{\mathcal{D}}\, b(\mathbf{x} - \mathbf{x}_j)\delta(s_i - \mathbf{n}(\theta_i)\cdot\mathbf{x})\,\mathrm{d}\mathbf{x}$.
Due to noise in the data or errors in the projection model the system of equations is inconsistent, so a solution may not exist. Furthermore, there may be many solutions that fit the observations equally well because the system is underdetermined. A standard approach to mitigate these issues is to formulate a regularized least-squares problem
\[
\min_{\mathbf{u}} {\textstyle\frac{1}{2}}\| W \mathbf{u} - \mathbf{p} \|_2^2 + {\textstyle\frac{\lambda}{2}}\|R\mathbf{u}\|_2^2,
\]
where $R$ is the regularization operator. Such a formulation is popular mainly because very efficient algorithms exist for solving it. Depending on the choice of $R$, however, this formulation forces the solution to have certain properties which may not reflect the truth. For example, setting $R$ to be the discrete Laplace operator will produce a smooth reconstruction, whereas setting $R$ to be the identity matrix forces the individual coefficients $u_i$ to be small. In many applications such quadratic regularization terms do not reflect the characteristics of the object we are reconstructing. For example, if we expect $u$ to be piecewise constant, we could use a Total Variation regularization term $\|R\mathbf{u}\|_1$ where $R$ is a discrete gradient operator \cite{sidky2008image}. Recently, a lot of progress has been made in developing efficient algorithms for solving such non-smooth optimization problems \cite{chambolle2011first}. If the object under investigation is known to consist of only two distinct materials, the regularization can be formulated in terms of a non-convex constraint $\mathbf{u} \in \{u_0, u_1\}^n$. The latter leads to a combinatorial optimization problem, solutions to which can be approximated using heuristic algorithms \cite{batenburg2011dart}.

In this paper, we consider tomographic reconstruction of \emph{partially discrete} objects that consist of a region of constant density embedded in a continuously varying background. In this case, neither the quadratic, Total Variation nor non-convex constraints by themselves are suitable. We therefore propose the following parametrization
\[
u(\mathbf{x})=
\left\{
\begin{matrix}
u_0(\mathbf{x})&\text{if}\,\,\mathbf{x}\in\Omega,\\
u_1&\text{otherwise}.\\
\end{matrix}
\right.
\]
The inverse problem now consists of finding $u_0(\mathbf{x})$, $u_1$ and the set $\Omega$.
We can subsequently apply suitable regularization to $u_0$ separately. To formulate a tractable optimization algorithm, we represent the set $\Omega$ using a level-set function $\phi(\mathbf{x})$ such that
\[
\Omega = \{\mathbf{x}\,|\, \phi(\mathbf{x}) > 0 \}.
\]
In the following sections, we discuss how to formulate a variational problem to reconstruct $\Omega$ and $u_0$ based on a parametric level-set representation of $\Omega$ and assuming we know $u_1$. 

The outline of the paper is as follows.
In section~\ref{section:levelset} we discuss the parametric level-set method and propose some practical heuristics for choosing various paramaters that occur in the formulation.
A joint background-anomaly reconstruction algorithm for partially discrete tomography is discussed in section \ref{section:jointrec}. The results on few moderately complicated numerical phantoms are presented in Section~\ref{section:results}. We provide some concluding remarks in Section~\ref{section:discussion}.

\section{Level-set methods}
\label{section:levelset}

In terms of the level-set function, we can express $u$ as 
\[
u(\mathbf{x}) = (1 - h(\phi(\mathbf{x})))u_0(\mathbf{x}) + h(\phi(\mathbf{x}))u_1,
\]
where $h$ is the Heaviside function and the latter term represents the anomaly.

Level-set methods have received much attention in geometric inverse problems, interface tracking, segmentation and shape optimization. The reason being their ability to handle topological changes. 
The classical level-set method, introduced by Sethian and Osher \cite{osher1988fronts}, solves the Hamiltonian-Jacobi equation, also known as level-set equation.
\begin{equation}
\frac{\partial \phi}{\partial t} + v | \nabla \phi | = 0,
\end{equation}
where $\phi: \mathbb{R}^2 \times \mathbb{R}^+ \rightarrow \mathbb{R}$ denotes the level-set function as a time-dependent quantity for representing the shape and $v$ denotes the normal velocity. In the inverse-problems setting, the velocity $v$ is often derived from the gradient of the cost function with respect to the model parameter \cite{burger2001level}, \cite{dorn2006level}. There are various numerical issues associated with the numerical solution of level-set equation, e.g. reinitialization of the level-set. We refer the interested reader to a seminal paper in level-set method \cite{osher2006level} and its application to computational tomography \cite{klann2011mumford}.

Instead of taking this classical level-set approach, we employ a parametric level-set approach, first introduced by Aghasi et al \cite{aghasi2011parametric}. In this method, the level-set function is parametrized using radial basis functions:
\[ 
\phi(\mathbf{x}) = \sum_{j=1}^{n'} \alpha_j \Psi ( \beta_j\|\mathbf{x} - \boldsymbol{\chi}_j\|_2 ), 
\]
where $\Psi(.)$ is a radial basis function, $\{\alpha_j\}_{j=1}^{n'} $ and $\{\chi_j\}_{j=1}^{n'}$ are the amplitudes and nodes respectively, and the parameters $\{\beta_j\}_{j=1}^{n'}$ control the widths. Introducing the kernel matrix $A(\boldsymbol{\chi}, \boldsymbol{\beta})$ with elements
\[
a_{ij} = \Psi(\beta_j\|\mathbf{x}_i - \boldsymbol{\chi}_j \|_2),
\]
we can now express $\mathbf{u}$ as
\begin{equation}
\mathbf{u} = (1 - h(A(\boldsymbol{\chi}, \boldsymbol{\beta}) \boldsymbol{\alpha}))\odot \mathbf{u}_0 + h(A(\boldsymbol{\chi}, \boldsymbol{\beta}) \boldsymbol{\alpha}) u_1,
\label{eq:getufrompls}
\end{equation}
where $h$ is applied element-wise to the vector $A(\boldsymbol{\chi}, \boldsymbol{\beta}) \boldsymbol{\alpha}$ and $\odot$ denotes the element-wise (Hadamard) product.
By choosing the parameters $(\boldsymbol{\chi}, \boldsymbol{\beta}, \boldsymbol{\alpha})$ appropriately we can represent any (smooth) shape. To simplify matters and make the resulting optimization problem more tractable, we consider a fixed regular grid $\{\boldsymbol{\chi}_j\}_{j=1}^{n'}$ and a fixed width $\beta_j \equiv \beta$. In the following we choose $\beta$ in accordance with the gridspacing $\Delta \chi$ as $\beta = 1/(\eta \Delta \chi)$, where $\eta$ determines the influence of RBF on its neighbors.

\subsubsection{Example}
To show that the reconstruction of level-set with a finitely many radial basis functions, we consider the level-set shown in Figure~\ref{fig:lsrbf} (a). With ${n'} = 196$ RBFs, it is possible to reconstruct a smooth shape discretized on a grid with $n = 256 \times 256$ pixels.
\begin{figure}[!ht]
\centering
\begin{tabular}{cccc}
(a) & (b) & (c) & (d) \\
 \includegraphics[width=0.2\columnwidth]{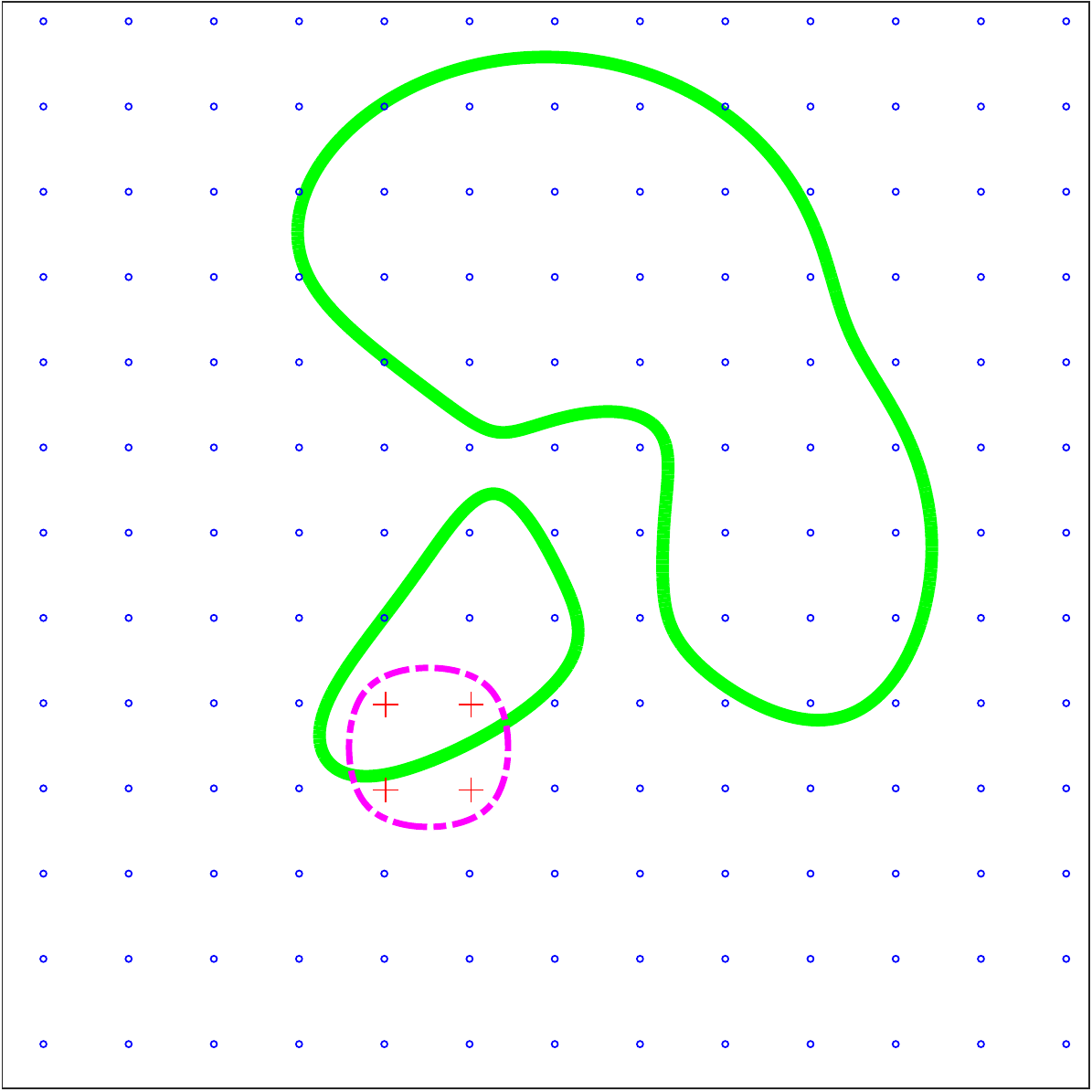} & \includegraphics[width=0.25\columnwidth]{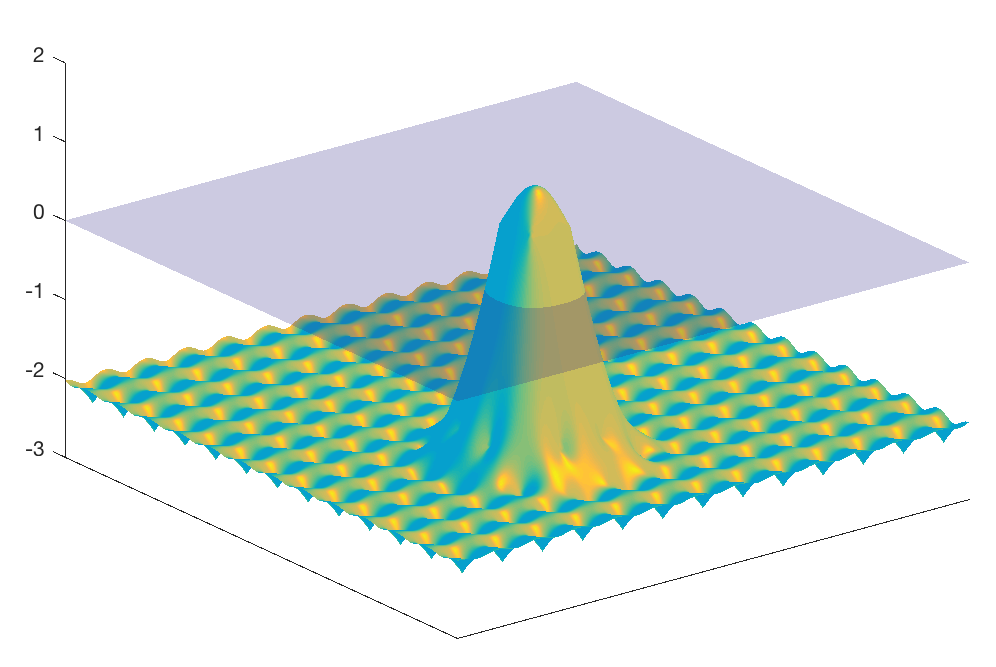} &
 \includegraphics[width=0.2\columnwidth]{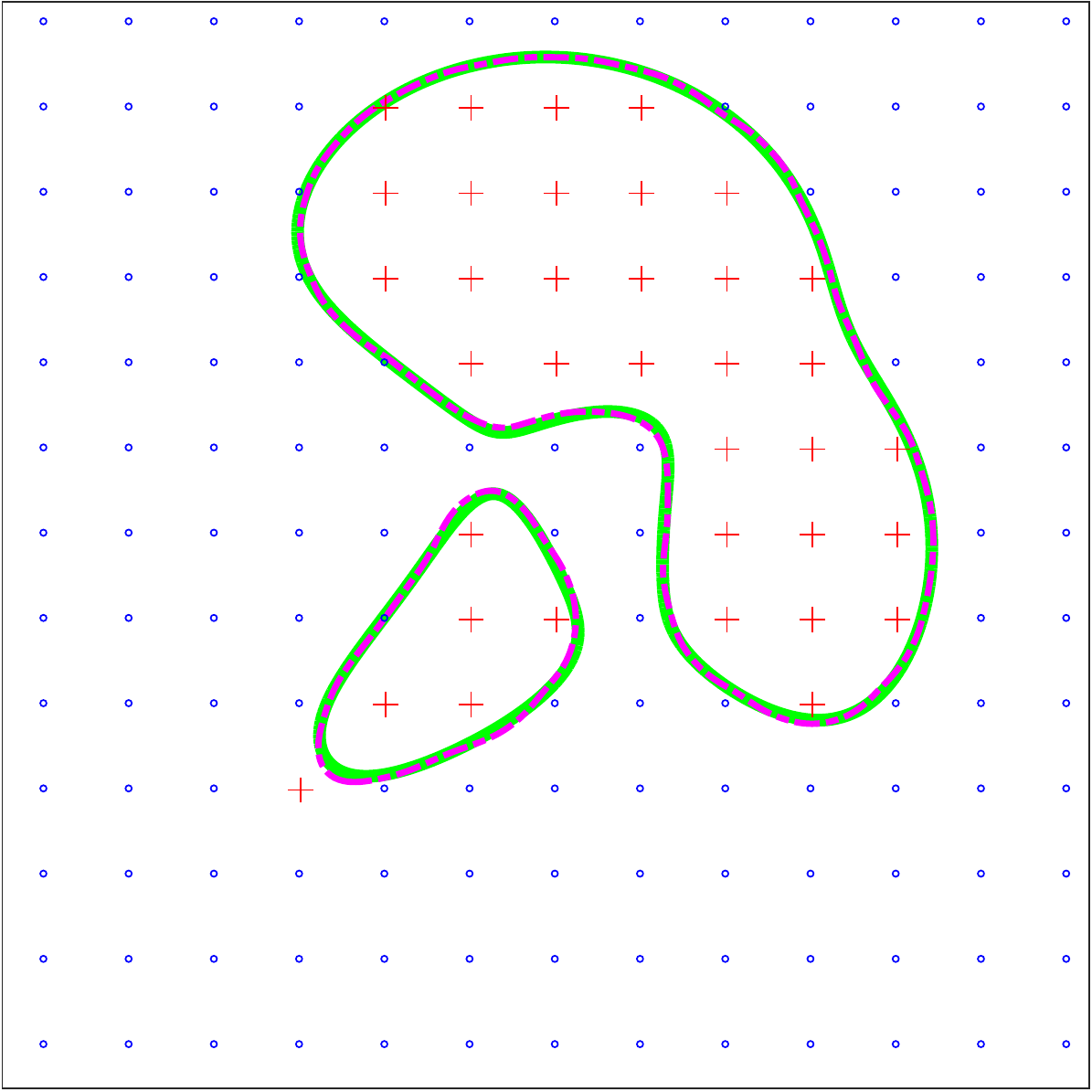} & \includegraphics[width=0.25\columnwidth]{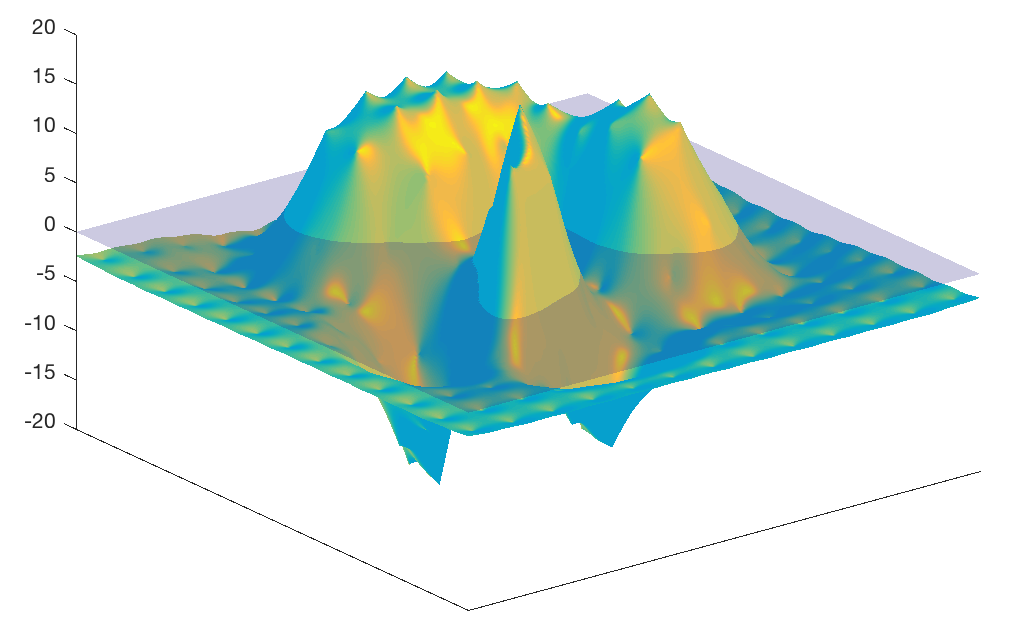}
\end{tabular} 
\caption{Any (sufficiently) smooth level-set can be reconstructed from radial basis functions. (a) Level-set to be reconstructed is denoted by \textit{green} line. Initial level-set (\textit{dash-dotted} line) is generated by some positive RBFs (denoted by \textit{red plusses}) near the center and negative RBFs all around (denoted by \textit{blue dots}) (b) Initial level-set function and the 0-level plane (c) Reconstructed level-set denoted by \textit{dash-dotted} line with corresponding positive and negative RBFs  (d) Final level-set function}
\label{fig:lsrbf}
\end{figure}

Finally, the discretized reconstruction problem for determining the shape is now formulated as 
\begin{align}
\min_{\boldsymbol{\alpha}} \left\lbrace f(\boldsymbol{\alpha}) = \| W[ (u_1 - \mathbf{u}_0)\odot h_{\epsilon}(A \boldsymbol{\alpha})] - (\mathbf{p} - W\mathbf{u}_0)\|_2^2 \right\rbrace,
\label{eq:PLSf}
\end{align} 
where $h_{\epsilon}$ is a smooth approximation of the Heaviside function. The gradient and Gauss-Newton Hessian of ${f(\boldsymbol{\alpha})}$ are given by 
\begin{align}
\label{eq:PLSgH}
\begin{split}
\nabla f(\boldsymbol{\alpha}) &= A^T D_{\boldsymbol{\alpha}}^T W^T \mathbf{r}(\boldsymbol{\alpha}),
\\
H_{GN}(f(\boldsymbol{\alpha})) &=  A^T D_{\boldsymbol{\alpha}}^T W^TW D_{\boldsymbol{\alpha}} A  .
\end{split}
\end{align}
where the diagonal matrix and residual vectors are given by
\begin{align*}
D_{\boldsymbol{\alpha}} = \diag{ (u_1-\mathbf{u}_0) \odot h_{\epsilon}'(A\boldsymbol{\alpha}) }, \quad & 
\mathbf{r}(\boldsymbol{\alpha}) = W[(u_1 - \mathbf{u}_0)\odot h_\epsilon(A \boldsymbol{\alpha})] - (\mathbf{p} - W\mathbf{u}_0).
\end{align*}
Using a Gauss-Newton method, the level-set parameters are updated as
\[
\boldsymbol{\alpha}^{(k+1)} = \boldsymbol{\alpha}^{(k)} - \mu^{(k)}\left(H_{GN}(f(\boldsymbol{\alpha}^{(k)}))\right)^{-1}\nabla f(\boldsymbol{\alpha}^{(k)}),
\]
where $\mu_k$ is a suitable stepsize and $\boldsymbol{\alpha}^{(0)}$ is a given initial estimate of the shape.

From equation~\ref{eq:PLSgH}, it can be observed that the ability to update the level-set parameters depends on two main factors: 1) The difference between $\mathbf{u}_0$ and $u_1$, and 2) the derivative of the Heaviside function. Hence, the support and smoothness of $h'_{\epsilon}$ plays a crucial role in the sensitivity. More details on the choice of $h_{\epsilon}$ are discussed in section \ref{subsection:heaviside}.

\subsubsection{Example}
We demonstrate the parametric level-set method on a (binary) discrete tomography problem. We consider the model described in Figure~\ref{fig:discreteTomo}(a). For a full-angle case ($0 \leq \theta \leq \pi$) with a large number of samples, Figure~\ref{fig:discreteTomo}(c) shows that it is possible to accurately reconstruct a complex shape.
\begin{figure}
\centering
\renewcommand{\arraystretch}{1.5}
\begin{tabular}{>{\centering}m{1.4in} >{\centering}m{1.4in} >{\centering\arraybackslash}m{1.4in}}
(a) & (b) & (c) \\
\includegraphics[width=0.23\columnwidth]{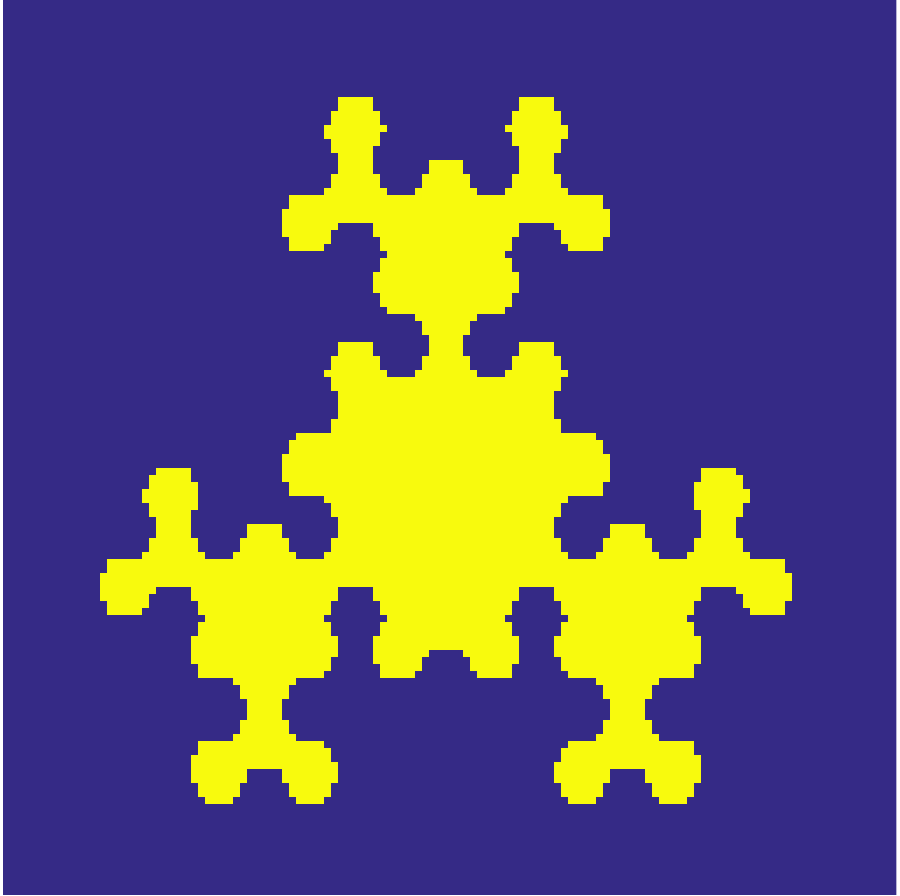}  & \includegraphics[width=0.25\columnwidth]{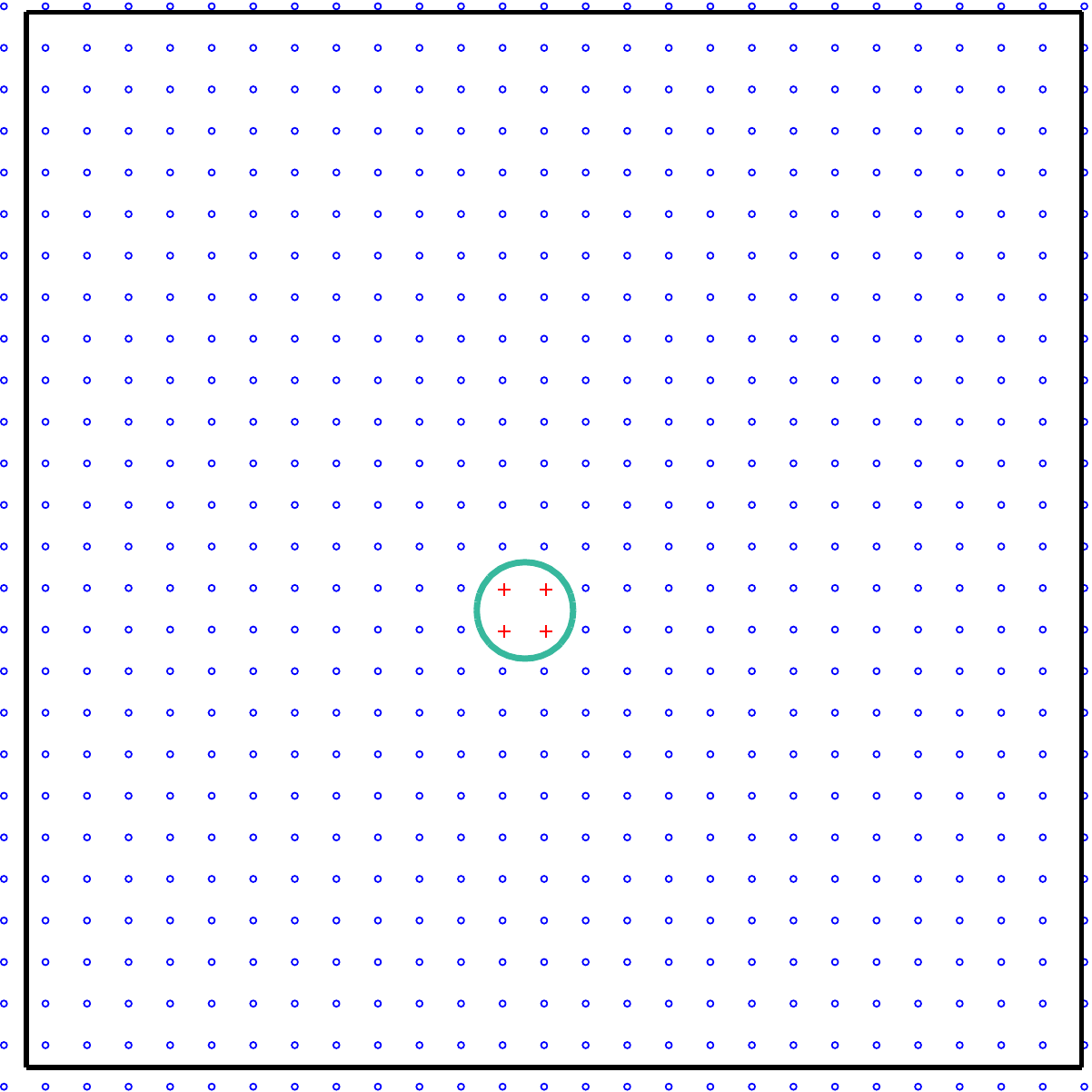} & \includegraphics[width=0.25\columnwidth]{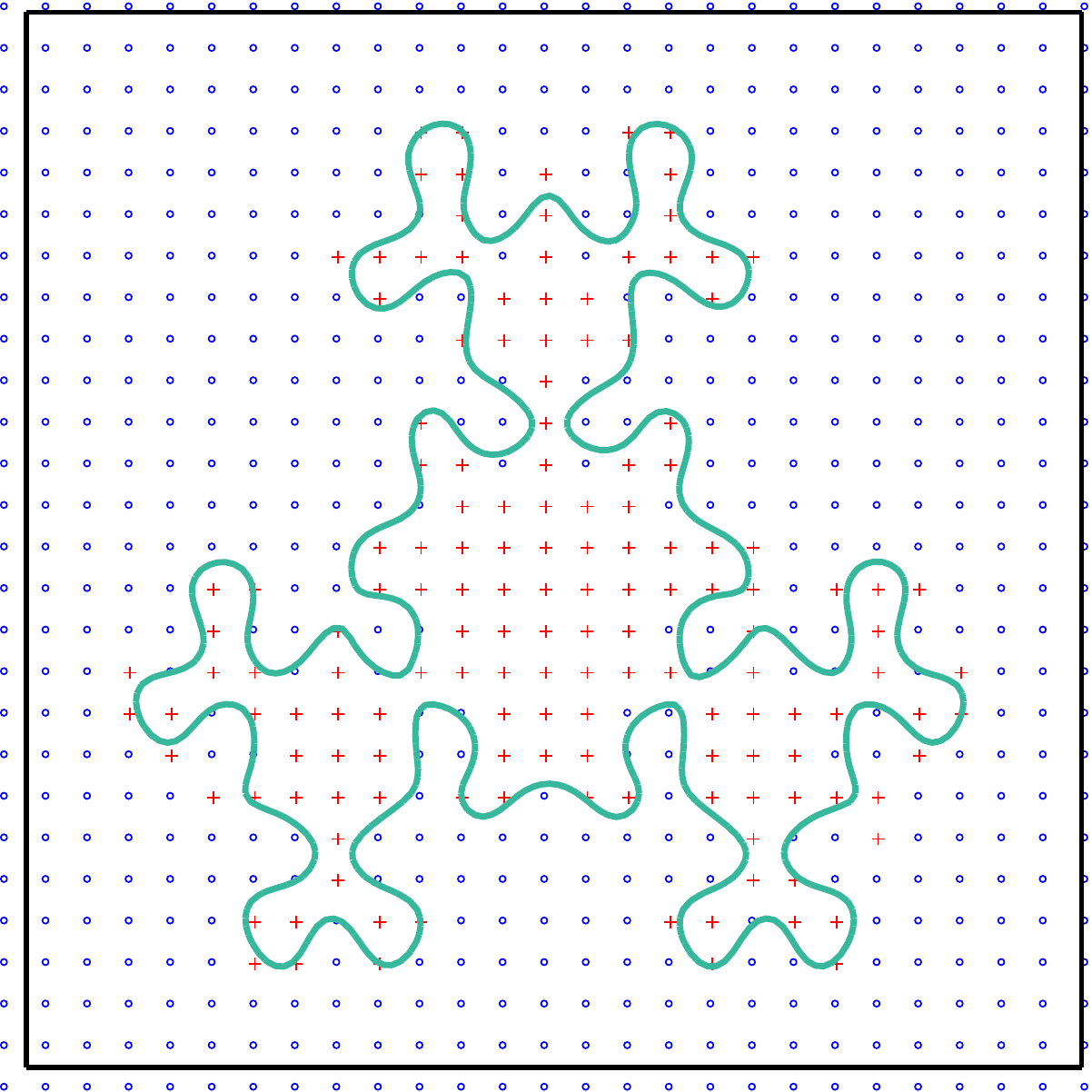}\\
(d) & (e) & (f) \\
\includegraphics[width=0.3\columnwidth]{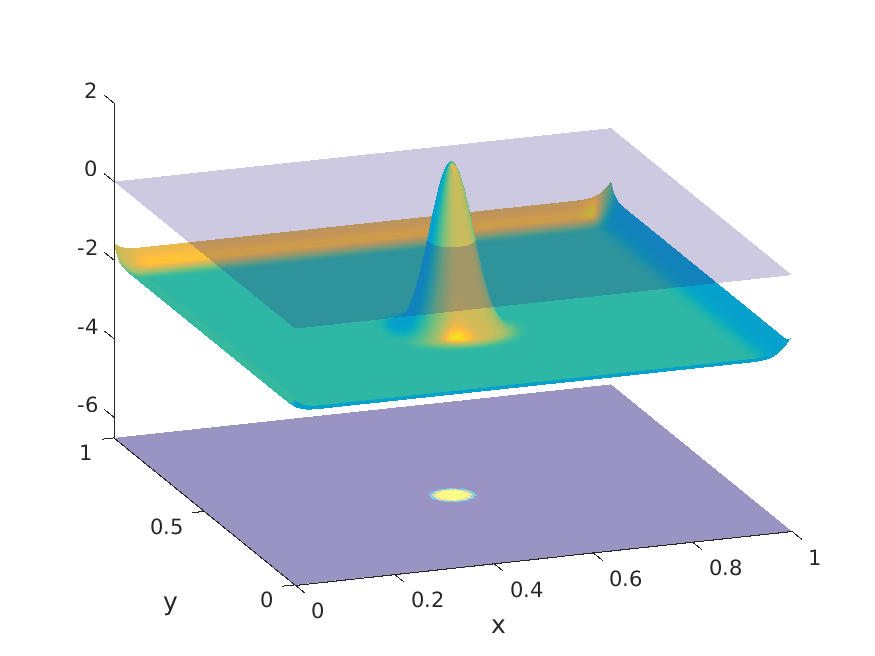} & \includegraphics[width=0.3\columnwidth]{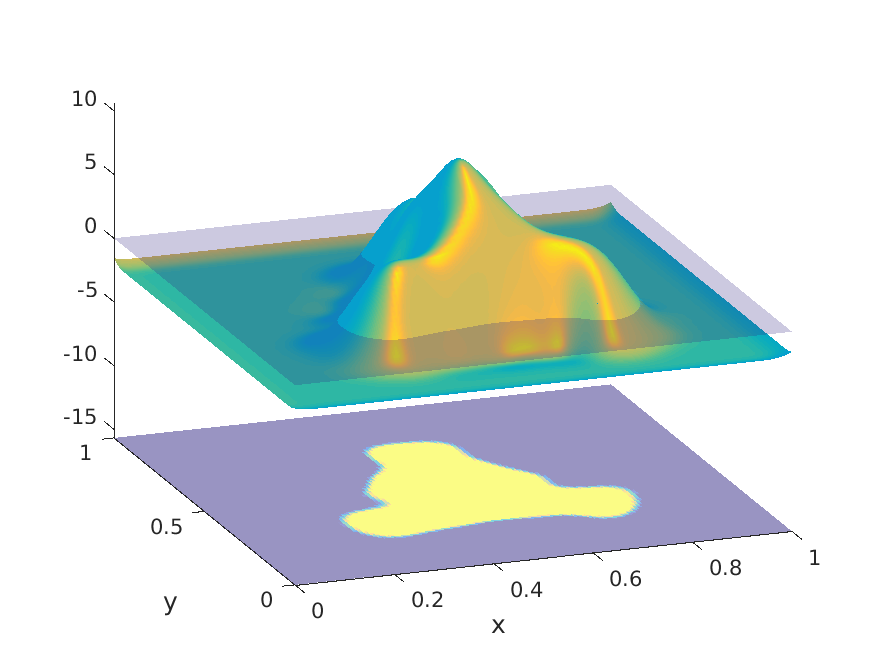} & \includegraphics[width=0.3\columnwidth]{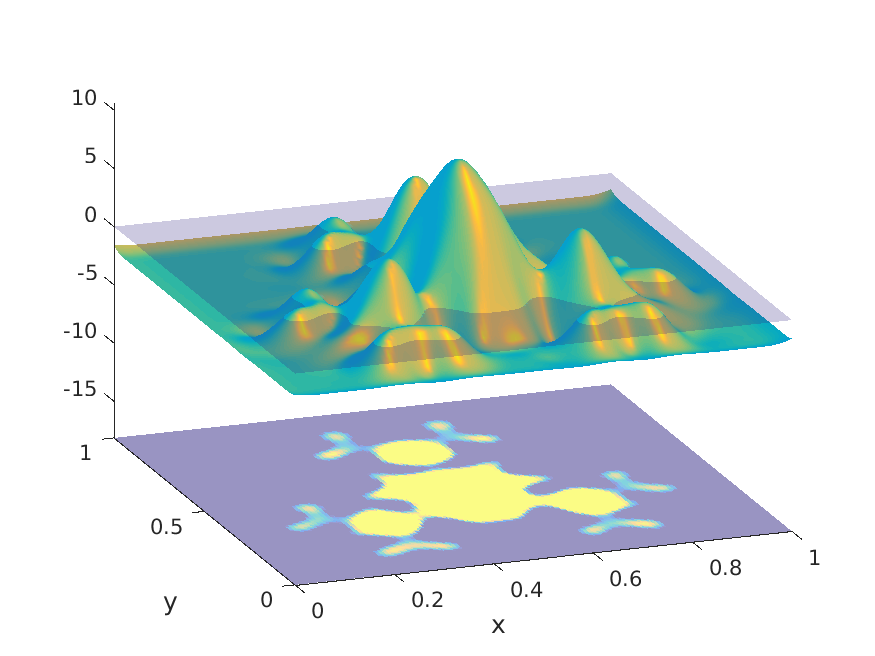}
\end{tabular} 
\caption{Parametric level-set method for Discrete tomography problem. (a) True model ($n = 256 \times 256$) (b) RBF grid ($n' = 27 \times 27$) with initial level-set denoted by \textit{green line}, positive and negative RBFs are denoted by \textit{red pluses} and \textit{blue dots} respectively  (c) Final level-set denoted by the \textit{green line}, and the corresponding positive and negative RBFs (d) Initial level-set function (e) level-set function after 10 iterations (f) final level-set function after 25 iterations.}
\label{fig:discreteTomo}
\end{figure}

\subsection{Approximation to Heaviside function}
\label{subsection:heaviside}
The update of the level-set function primarily depends on the Heaviside function. Various approximations have been mentioned earlier \cite{aghasi2011parametric}. These approximations suffer from the variational region of Dirac-Delta function near its peak ($\delta|_{x = 0}$) which amplifies the gradient disproportionally. This sometimes results in poor updates for the level-set parameter $\boldsymbol{\alpha}$, and hence ruining the reconstructions. To solve this issue, we propose a new formulation of the Heaviside function. We construct the piecewise Dirac-Delta function shown in equation~\eqref{eq:dirac-delta}:
\begin{align}
\delta(\mathbf{x}) = \begin{cases}
0 & \quad \mathbf{x} \leq -\epsilon \\
\frac{1}{4(1-\mu)\epsilon} \left( 1 + \frac{\mathbf{x}+(1-\mu)\epsilon}{\mu \epsilon} + \frac{1}{\pi}\sin(\pi \frac{\mathbf{x}+(1-\mu)\epsilon}{\mu \epsilon}) \right)   &  \quad   -\epsilon < \mathbf{x} \leq - \mu \epsilon \\
\frac{1}{2(1-\mu)\epsilon} &  \quad - \mu \epsilon < \mathbf{x} \leq \mu \epsilon \\
\frac{1}{4(1-\mu)\epsilon} \left( 1 - \frac{\mathbf{x}-(1-\mu)\epsilon}{\mu \epsilon} - \frac{1}{\pi}\sin(\pi \frac{\mathbf{x}-(1-\mu)\epsilon}{\mu \epsilon}) \right)   &  \quad  \mu \epsilon < \mathbf{x} \leq \epsilon \\
0 &  \quad   \mathbf{x} \geq \epsilon
\end{cases}
\label{eq:dirac-delta}
\end{align}
This new approximation has been plotted in Figure~\ref{fig:heavi}. The above formulation provides mainly 3 benefits: 1) constant sensitivity in the boundary region controlled by parameter $\mu$, 2) a smooth transition part and 3) the compact support. 

\begin{definition}
In accordance with the compact approximation of the Heaviside function with width $\epsilon$, a level-set boundary, denoted by $\partial \Omega$, is defined as the set of all the points $\mathbf{x} \in \mathbb{R}^2$ satisfying the condition $h_\epsilon'(\phi(\mathbf{x})) > 0$.
\end{definition}
Figure~\ref{fig:heavi}(c) shows a graphical representation of level-set boundary.
\begin{lemma}
For \textit{any} smooth and compact approximation of the Heaviside function with \textit{finite} width $\epsilon$, there exists a relation between level-set boundary and gradient of level-set function, given by $ |\delta_{\mathbf{x}}^T \nabla \phi(\mathbf{x})| \leq \epsilon $, where, $\delta_{\mathbf{x}} = \max_{x \in \partial \Omega} |\mathbf{x} - \mathbf{x}_0|$ and $\mathbf{x}_0$ is the point on the level-set.
\label{lemma:relation}
\end{lemma} 

\begin{proof}
From Taylor series expansion for $\phi(\mathbf{x})$ near the level-set point $\mathbf{x}_0$, we get
\[\phi(\mathbf{x}) = \phi(\mathbf{x}_0) + (\mathbf{x} - \mathbf{x}_0)^T \nabla \phi(\mathbf{x}_0) + \mathcal{O}(\|\mathbf{x} - \mathbf{x}_0\|^2) .\]
$h_\epsilon' (\phi(\mathbf{x})) > 0$ if and only if $| \phi(\mathbf{x}) | < \epsilon$. Neglecting higher-order terms, we get $ |(\mathbf{x} - \mathbf{x}_0)^T \nabla \phi(\mathbf{x}_0) | \leq \epsilon$. This implies the above relation.

\end{proof}

From the lemma~\ref{lemma:relation}, it is important to choose the Heaviside width in such a way that the level-set boundary exists on model grid. For simplicity, we crudely approximate the gradient of level-set function using upper and lower bounds \cite{kadu2016salt}. 
Hence, the heaviside width is represented by
\begin{equation}
\epsilon = \kappa \left( \frac{\max( \phi(\mathbf{x}) ) - \min (\phi (\mathbf{x}) )}{\Delta x} \right) = \kappa \left( \frac{\max(A \boldsymbol{\alpha} ) - \min (A \boldsymbol{\alpha})}{\Delta x} \right),	\label{eq:epsilon}
\end{equation}
where $\kappa$ controls the number of gridpoints a level-set boundary can have.
This formulation of $\epsilon$ solves the re-initialization issue associated with the level-set method. The steepness ($|\nabla \phi(\mathbf{x})| \gg 1$) of the level-set function in the level-set boundary can be handled by this formulation as well, as it adapts the level-set boundary to global change in level-set function.

\begin{figure}[!ht]
\centering
\renewcommand{\arraystretch}{1.5}
\begin{tabular}{ccc}
\includegraphics[scale=0.25]{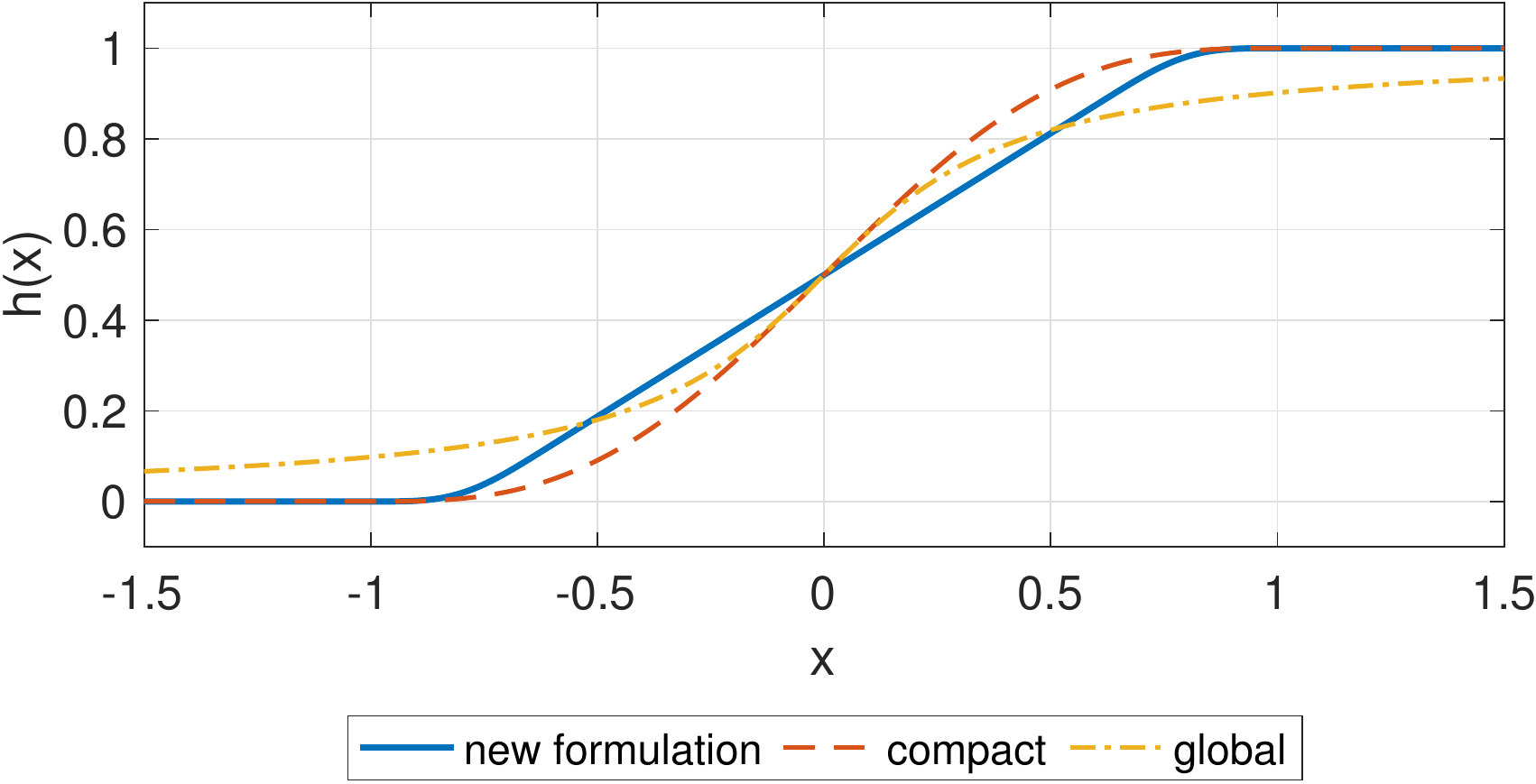} &  \includegraphics[scale=0.25]{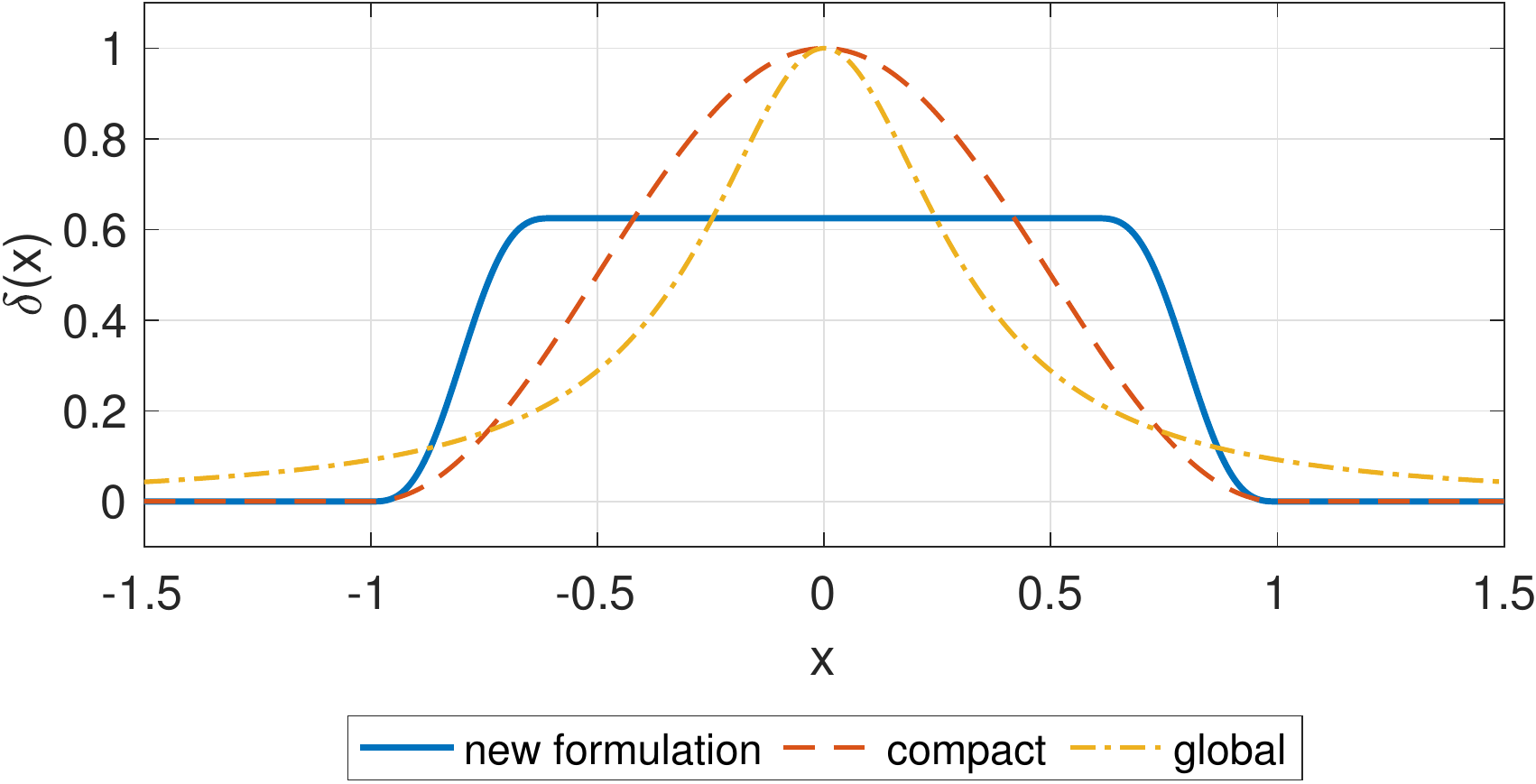} &  \includegraphics[scale=0.25]{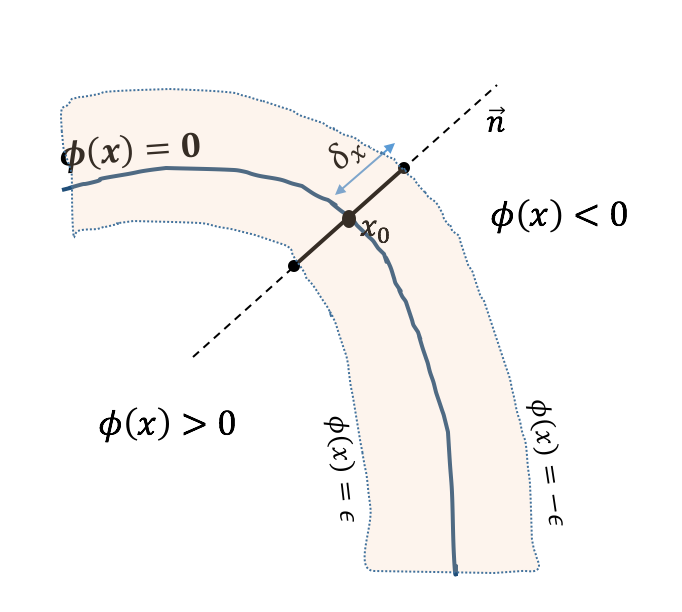}\\
(a) & (b) & (c)
\end{tabular}
\caption{New formulation for approximating the Heaviside function. the Heaviside functions (a) and corresponding Dirac-Delta functions (b) with $\epsilon = 1$ and $\mu = 0.2$ . Global approximation is constructed from inverse tangent function, while compact one is composed of linear and sinusoid functions. (c) level-set boundary (\textit{orange} region) around zero level-set denoted by \textit{blue} line, $n$ represents the normal direction at $\mathbf{x}_0$.}
\label{fig:heavi}
\end{figure}

\section{Joint reconstruction algorithm}
\label{section:jointrec}
Reconstructing both the shape and the background parameter can be cast as a bi-level optimization problem
\begin{equation}
\min_{\mathbf{u}_0, \boldsymbol{\alpha} } \left\lbrace f(\boldsymbol{\alpha},\boldsymbol{u}_0) := \tfrac{1}{2} \| W[(1-h(A\boldsymbol{\alpha}) \mathbf{u}_0 + h(A\boldsymbol{\alpha})u_1 ] - \mathbf{p} \|_2^2 + \tfrac{\lambda}{2} \|L \mathbf{u}_0 \|_2^2 \right\rbrace,
\label{eq:bilevel}
\end{equation}
where $L$ is of form $[L_x^T \quad L_y^T]^T$. $L_x$ and $L_y$ are the second-order finite-difference operators in $x$ and $y$ directions respectively. This optimization problem is \emph{separable}; it is quadratic in $\mathbf{u}_0$ and non-linear in $\boldsymbol{\alpha}$. In order to exploit the fact that the problem has a closed-form solution in $\mathbf{u}_0$ for each $\boldsymbol{\alpha}$, we introduce a reduced objective
\[
\overline{f}(\boldsymbol{\alpha}) = \min_{\mathbf{u}_0} f(\boldsymbol{\alpha},\mathbf{u}_0).
\]
The gradient and Hessian of this reduced objective are given by
\begin{eqnarray}
\nabla \overline{f}(\boldsymbol{\alpha}) &=& \nabla_{\boldsymbol{\alpha}} f(\boldsymbol{\alpha},\overline{\mathbf{u}}_0),\\
\nabla^2 \overline{f}(\boldsymbol{\alpha}) &=& \nabla^2_{\boldsymbol{\alpha}}f - \nabla^2_{\boldsymbol{\alpha},\mathbf{u}_0}f\left(\nabla^2_{\mathbf{u}_0} f\right)^{-1}\nabla^2_{\boldsymbol{\alpha},\mathbf{u}_0}f,
\end{eqnarray}
where $\overline{\mathbf{u}}_0 = \argmin_{\mathbf{u}_0} f(\boldsymbol{\alpha},\mathbf{u}_0)$ \cite{aravkin2012estimating}.

Using a modified Gauss-Newton algorithm to find a minimizer of $\overline{f}$, leads to the following alternating algorithm
\begin{eqnarray}
\mathbf{u}_0^{(k+1)} &=& \underset{\mathbf{u}_0}{\text{arg min}} f(\boldsymbol{\alpha}^{(k)},\mathbf{u}_0) \label{eq:computeu0} \\
\boldsymbol{\alpha}^{(k+1)} &=& \boldsymbol{\alpha}^{(k)} - \mu^{(k)}\left(H_{GN}(f(\boldsymbol{\alpha}^{(k)}))\right)^{-1}\nabla_{\boldsymbol{\alpha}}f(\boldsymbol{\alpha}^{(k)},\mathbf{u}_0^{(k+1)}), \label{eq:computealpha}
\end{eqnarray}
where the expressions for the gradient and Gauss-Newton Hessian are given by \eqref{eq:PLSgH}. 
Convergence of this alternating approach to a local minimum of \eqref{eq:bilevel} is guaranteed as long as the step-length satisfies the strong Wolfe conditions \cite{wright1999numerical}.

The reconstruction algorithm based on this iterative scheme is presented in Algorithm \ref{alg:joint}.

\begin{algorithm}[H]
 \caption{Joint Reconstruction Algorithm}
 \label{alg:joint}
 \begin{algorithmic}[1]
 \REQUIRE
 $\mathbf{p}$ - data, $W$ - forward modeling operator, $u_1$ - anomaly property, $A$ - RBF Kernel matrix, $\boldsymbol{\alpha}_0$ - initial RBF weights, $\kappa$ - Heaviside width parameter, $\mu$ - Heaviside inclination parameter
 \ENSURE $\boldsymbol{\alpha}_{K-1}$ - final weights, $\mathbf{u}$ - corresponding model
 \FOR {$k = 0$ to $K-1$ }
 \STATE compute Heaviside $\epsilon$ from equation~\eqref{eq:epsilon}
 \STATE compute background parameter $\mathbf{u}_0^{(k+1)}$ by solving equation~\eqref{eq:computeu0}  \label{alg:joint:backg}
 \STATE compute level-set parameter $\boldsymbol{\alpha}^{(k+1)}$ from equation~\eqref{eq:computealpha} \label{alg:joint:alpha}
 \ENDFOR
 \STATE compute $\mathbf{u}$ from equation~\eqref{eq:getufrompls}.
 \end{algorithmic}
\end{algorithm}
\vspace{-5mm}
We use the LSQR method in step~\ref{alg:joint:backg}, with pre-defined maximum iterations and a tolerance value. A trust-region method is applied to compute $\boldsymbol{\alpha}^{(k+1)}$ in step~\ref{alg:joint:alpha} restricting the conjugate gradient to \textit{only} 10 iterations.

\section{Numerical Experiments}
\label{section:results}
The numerical experiments are performed on 4 phantoms shown in figure~\ref{fig:phantoms}. Each phantom has a constant gray value of parameter 1. For the first two phantoms, the background  varies from 0 to 0.5, while for the next two, it varies from 0 to 0.8. In order to avoid \textit{inverse crime}, the data is generated using a line Kernel, and the forward model uses a Joseph kernel. We use ASTRA toolbox to compute the forward and backward projections \cite{bleichrodt2016easy}.  First, we show the results on the noiseless full-view data and later we compare various methods to proposed method in limited-data case with additive gaussian noise of 10~dB SNR.

\begin{figure}[!ht]
\centering
\renewcommand{\arraystretch}{1.5}
\begin{tabular}{cccc} 
\includegraphics[scale=0.25]{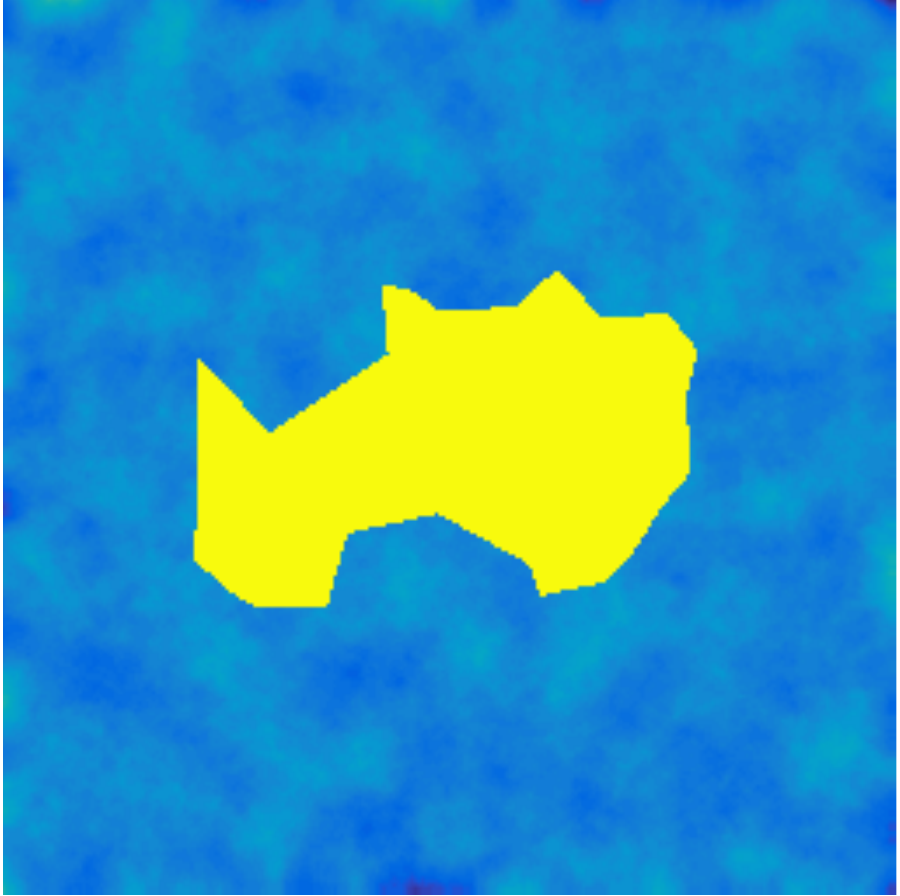} & \includegraphics[scale=0.25]{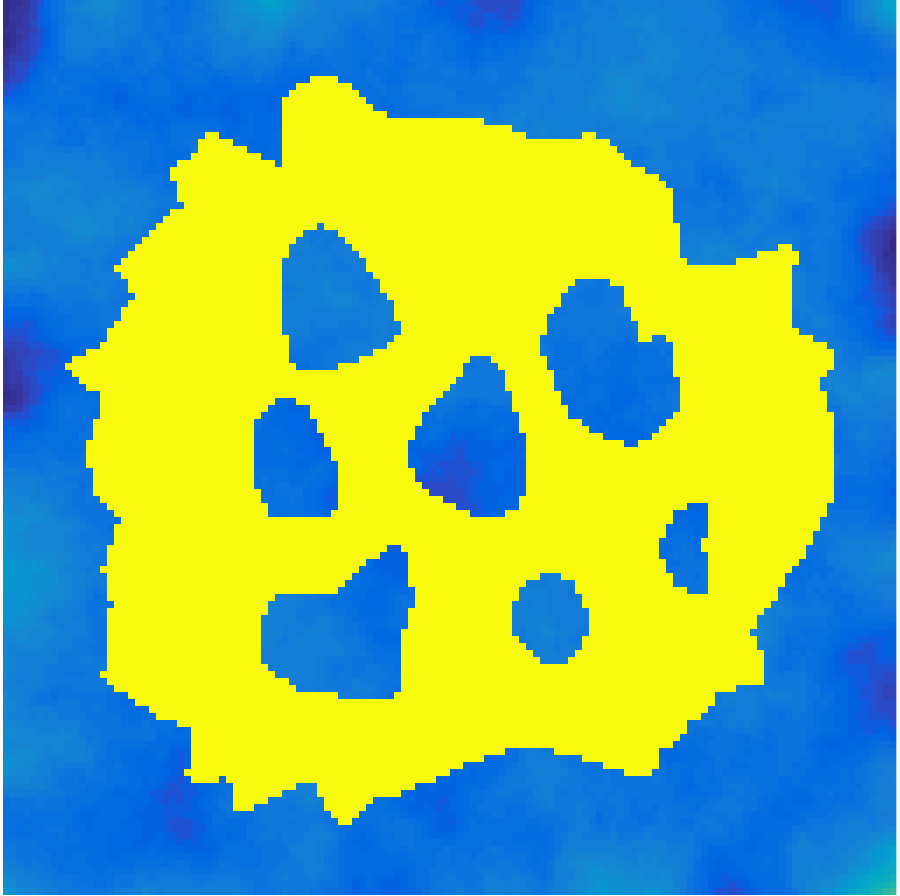} & \includegraphics[scale=0.25]{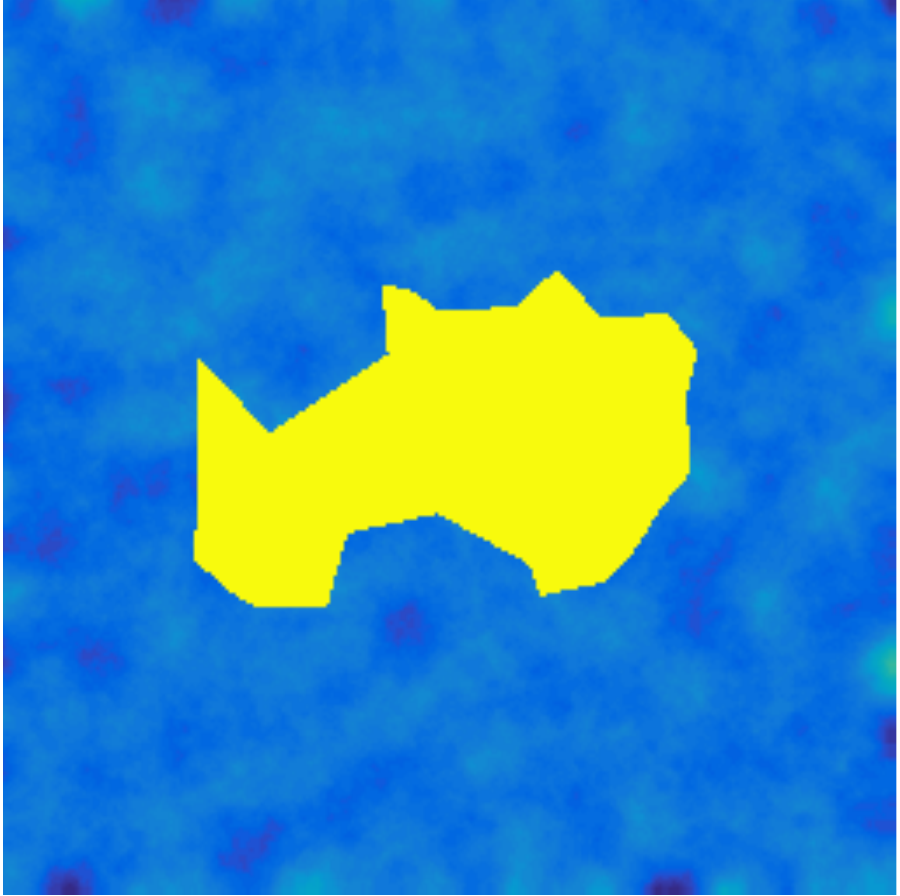} &
\includegraphics[scale=0.25]{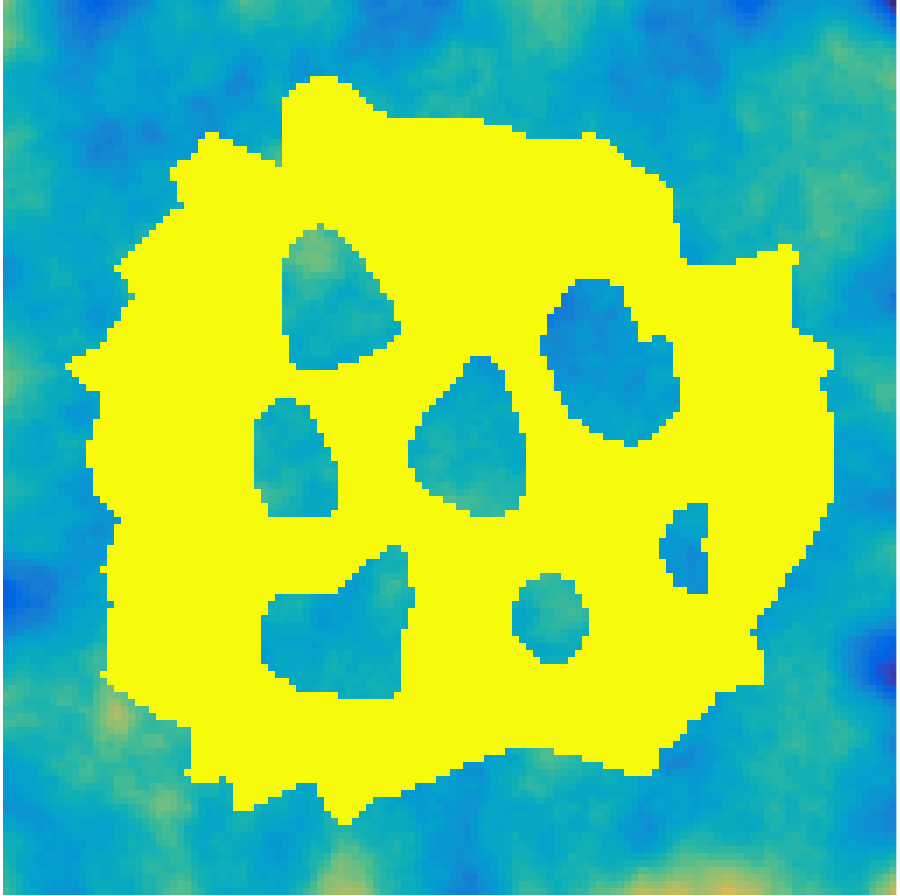} \\
(a) Model A &  (b) Model B &  (c) Model C & (d) Model D
\end{tabular}
\caption{Phantoms for Simulations. All the models have resolution of $256 \times 256$ pixels.}
\label{fig:phantoms}
\end{figure}

For the parametric level-set method, we use compactly supported radial basis functions. The basis functions has the form given below:
\[ \Psi(r) = (1 - r)_+^8 (32r^3 + 25r^2 + 8r + 1).\]

RBF nodes are placed on a rectangular grid with the gridspacing 5 times the computational (model) gridspacing. The grid extends to two points outside the model grid to compensate for the background effects. The heaviside width parameter $\kappa$ is set to be 0.01 and the its inclination parameter $\mu$ is set to be 0.1.

The level-set parameter $\boldsymbol{\alpha}$ is optimized using the \textit{fminunc} package (trust-region algorithm) in MATLAB. A total of 50 iterations are performed for predicting the $\boldsymbol{\alpha}$, while 200 iterations are performed for predicting $\mathbf{u}_0(x)$ using LSQR at each step.

\subsection{Regularization parameter selection}
The reconstruction with the proposed algorithm is influenced by the parameter for Tikhonov regularization. In general, there are various strategies to choose this parameter, e.g., \cite{thompson1991study}. As our problem formulation deals with the non-linearity in the level-set parameter, application of these kinds of strategies is not clear. Instead we analyze the various residuals, introduced below, with respect to the regularization parameter.

We define three measures (all in the least-squares sense) to quantify the residuals: 1) data residual (DR), determines the data fit between the true data and reconstructed data, 2) model residual (MR), determines the fit between reconstructed model and true model, 3) shape residual (SR), determines the fit between the reconstructed and true anomaly shape. In practice, one can only have a data residual measure to figure out the regularization parameter $\lambda$. From Figure~\ref{fig:regParam}, it is evident that there exists a sufficient region of $\lambda$ for which the reconstructions almost stays constant. This region is easily identifiable from the data residual plot for various $\lambda$.  

\begin{figure}[!ht]
\centering
\renewcommand{\arraystretch}{1}
\begin{tabular}{c}
(a) \\
\includegraphics[width=0.54\textwidth]{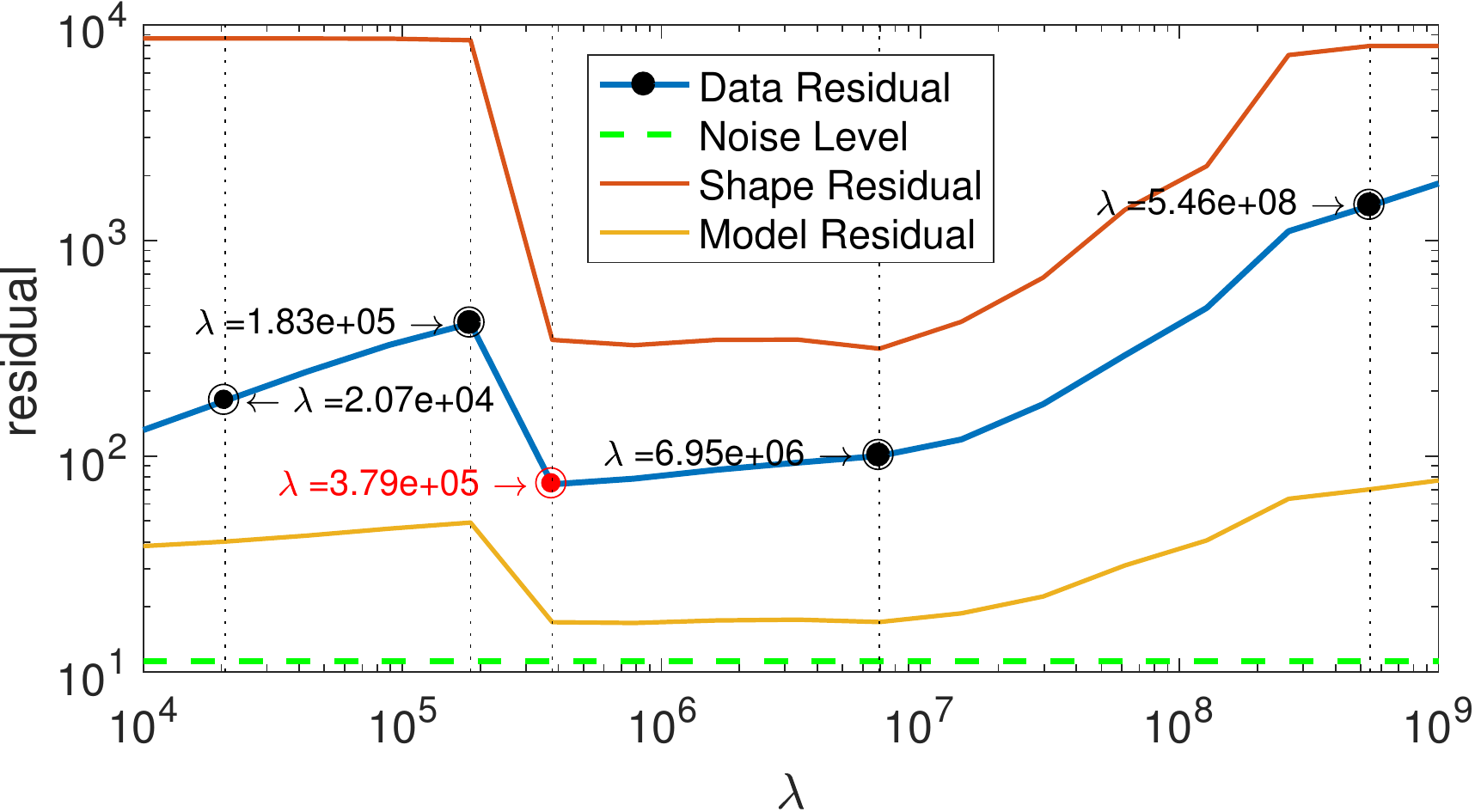}
\end{tabular}
\begin{tabular}{c c}
(b)$\lambda = 1.83 \times 10^5$ & (c)$\lambda = 3.79 \times 10^5$\\ 
\includegraphics[width=0.15\textwidth]{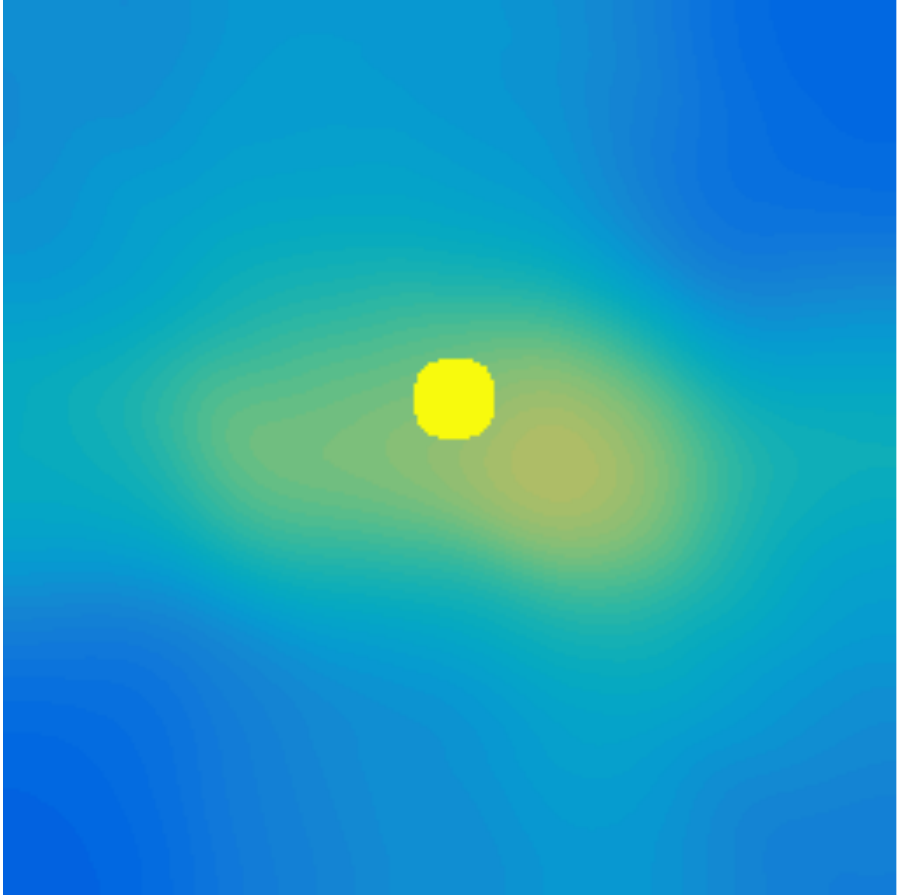} & \includegraphics[width=0.15\textwidth]{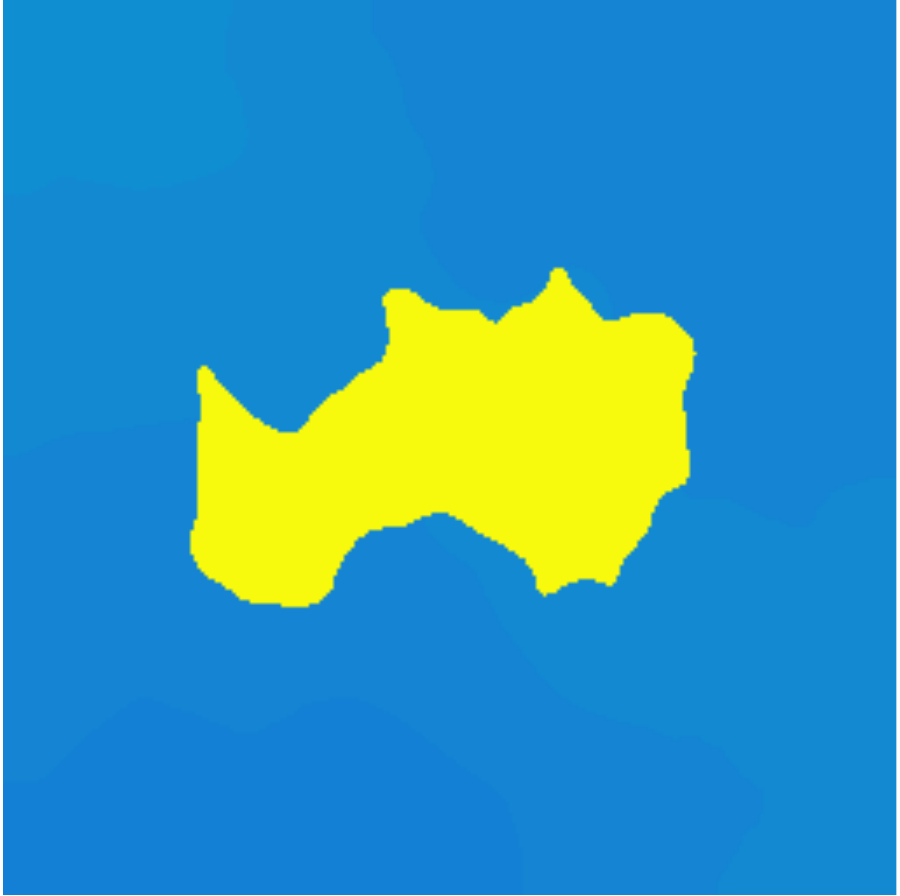} \\
(d) $\lambda = 6.95 \times 10^6$ & (e) $\lambda = 5.46 \times 10^8$ \\
\includegraphics[width=0.15\textwidth]{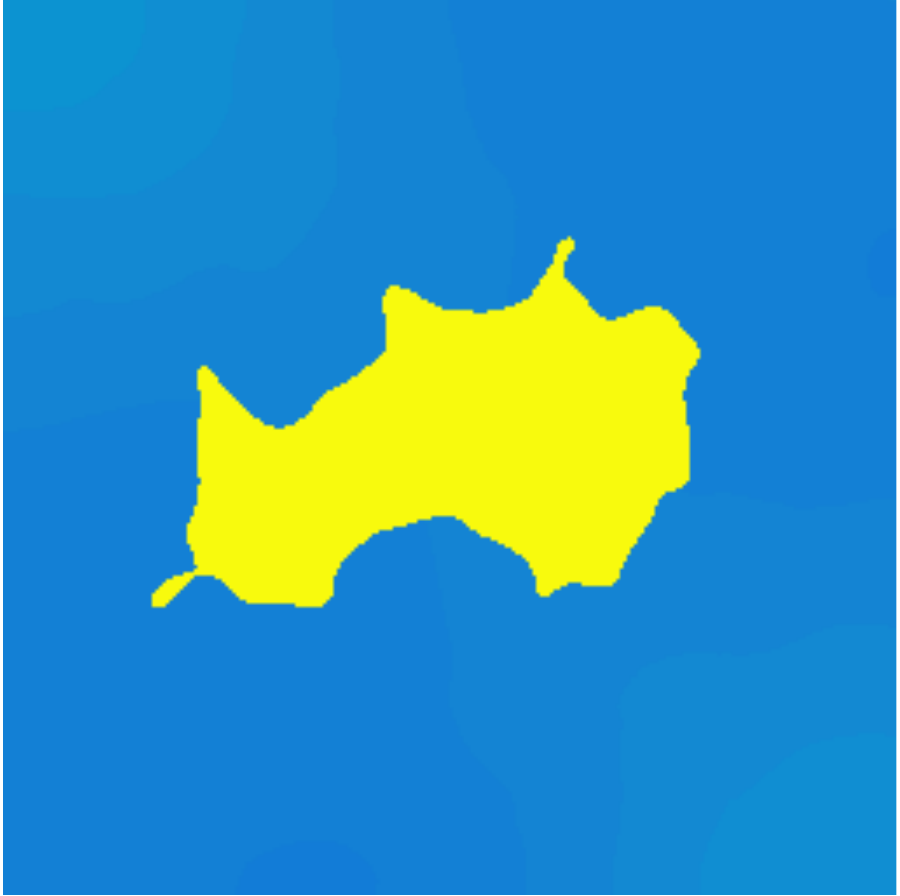} & \includegraphics[width=0.15\textwidth]{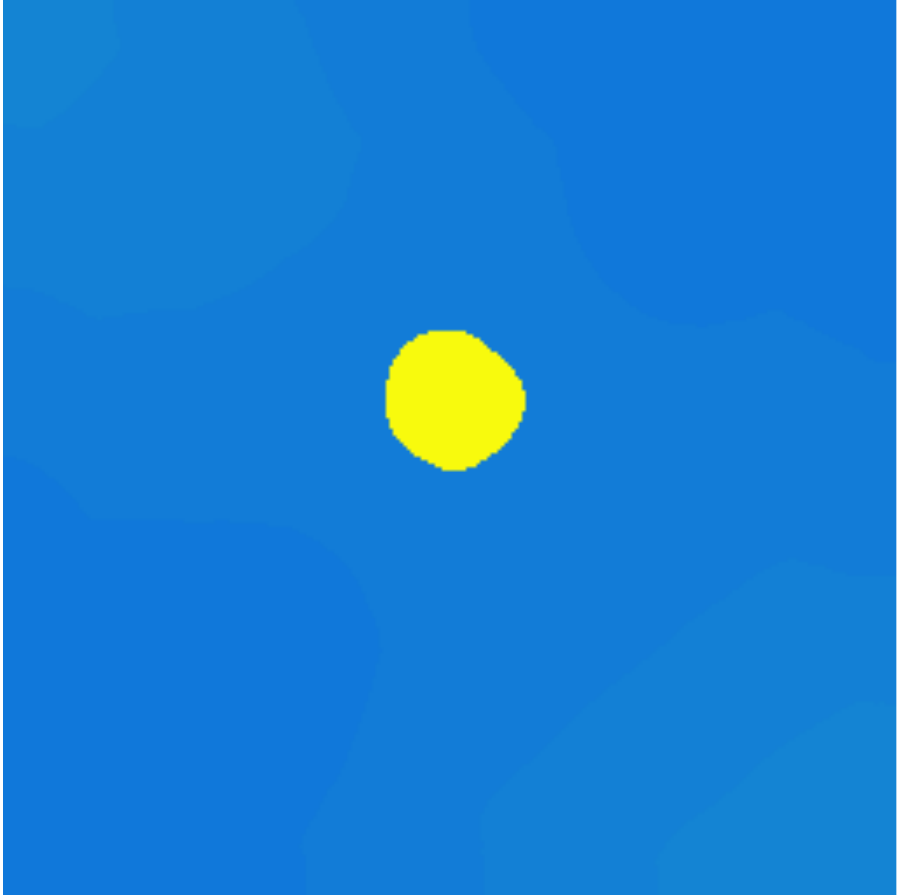}
\end{tabular}
\caption{Variation of residuals with regularization parameter for Tikhonov. Appropriate region for chosing $\lambda$ exists between $3.79 \times 10^5$ and $6.95 \times 10^6$. (a) behavior of DR, MR and SR over $\lambda$ for model A with noisy limited-angle data. Noise amplitude is denoted by \textit{green dotted line}. (b),(c),(d),(e) shows reconstructions for various $\lambda$ values}
\label{fig:regParam}
\end{figure}

\subsection{Benchmark test}
For the full-view (benchmark) case, the projection data is generated on $256 \times 256$ grid with 256 detectors and 180 projections with $0 \leq  \theta \leq \pi$. The noise is assumed to be zero in this case.
The results on the phantoms with the full-view data are shown in Figure~\ref{fig:benchmark}. Anomaly geometries in all of these models are reconstructed \textit{almost} perfectly with the proposed method, although the background has been smoothened out with the tikhonov regularization.
\begin{figure}[!ht]
\centering
\renewcommand{\arraystretch}{1}
\begin{tabular}{ccccccc} 
\includegraphics[scale=0.25]{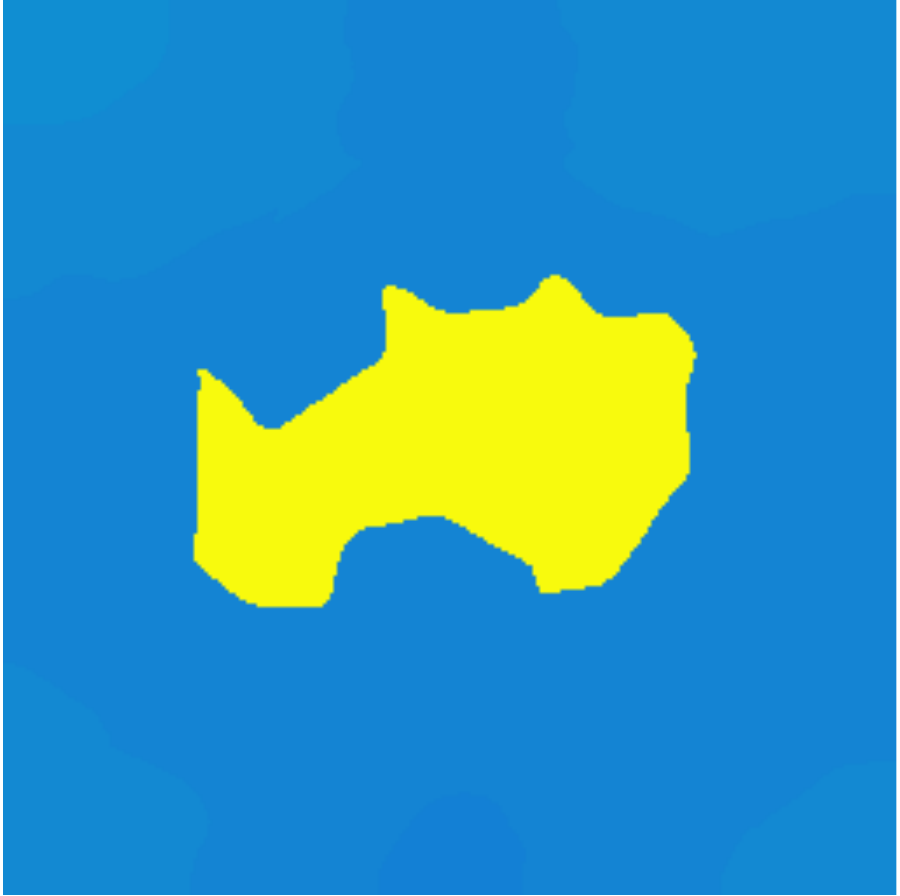} & & \includegraphics[scale=0.25]{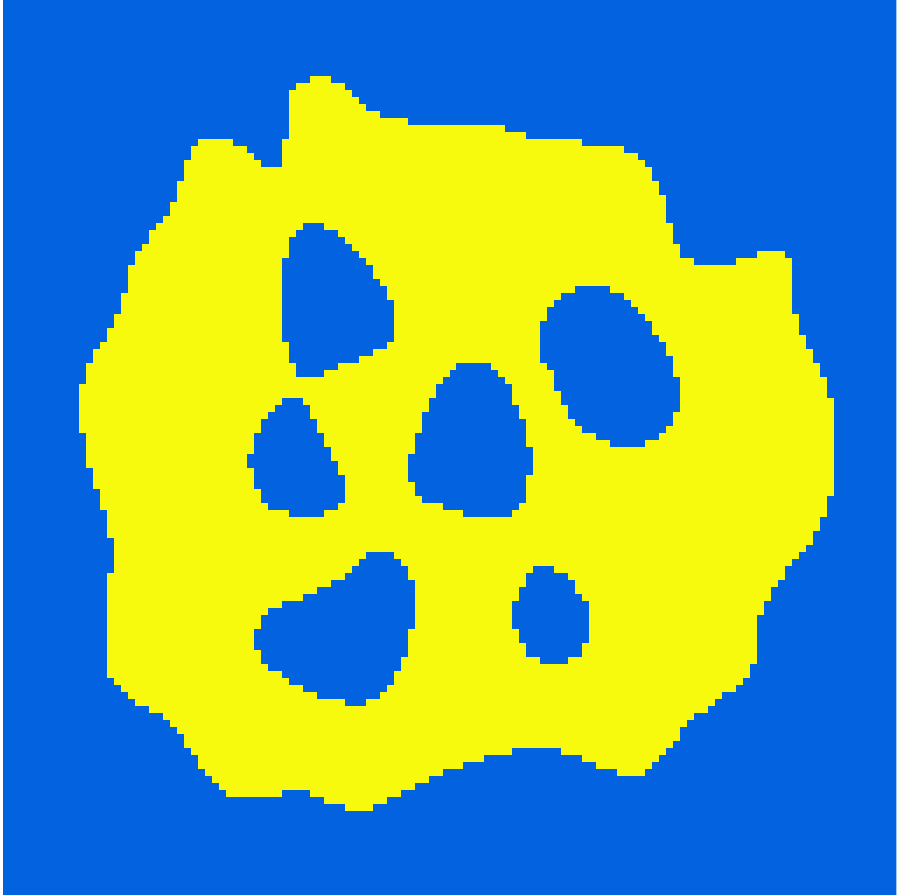} & & \includegraphics[scale=0.25]{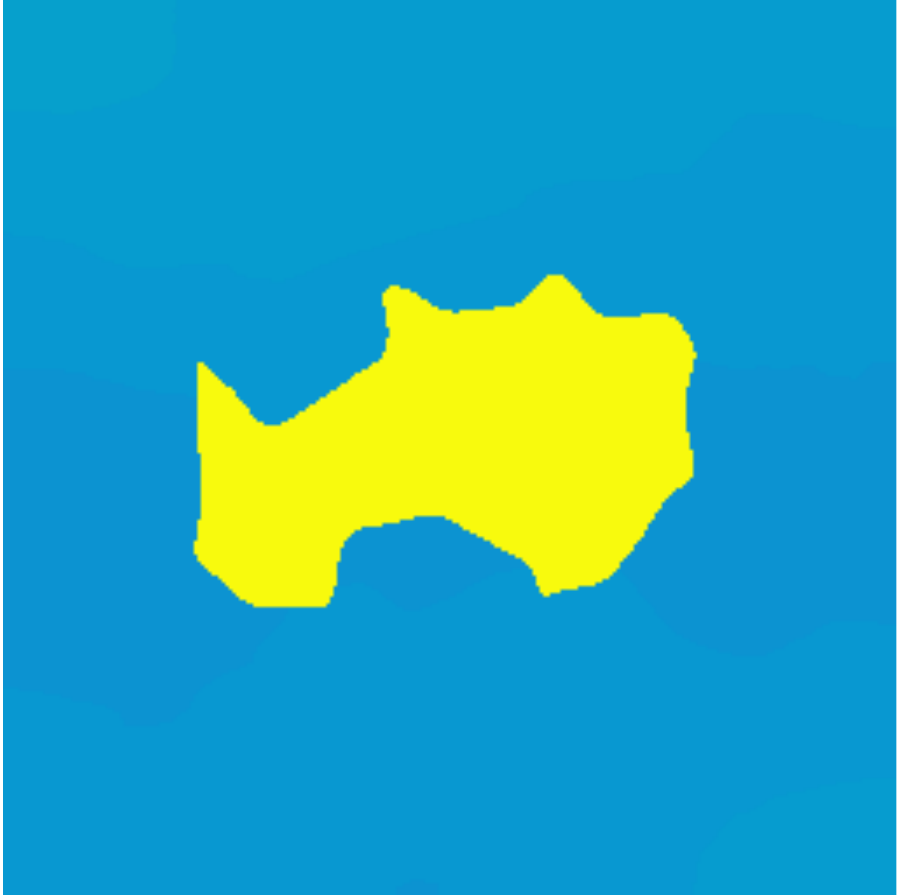} & &\includegraphics[scale=0.25]{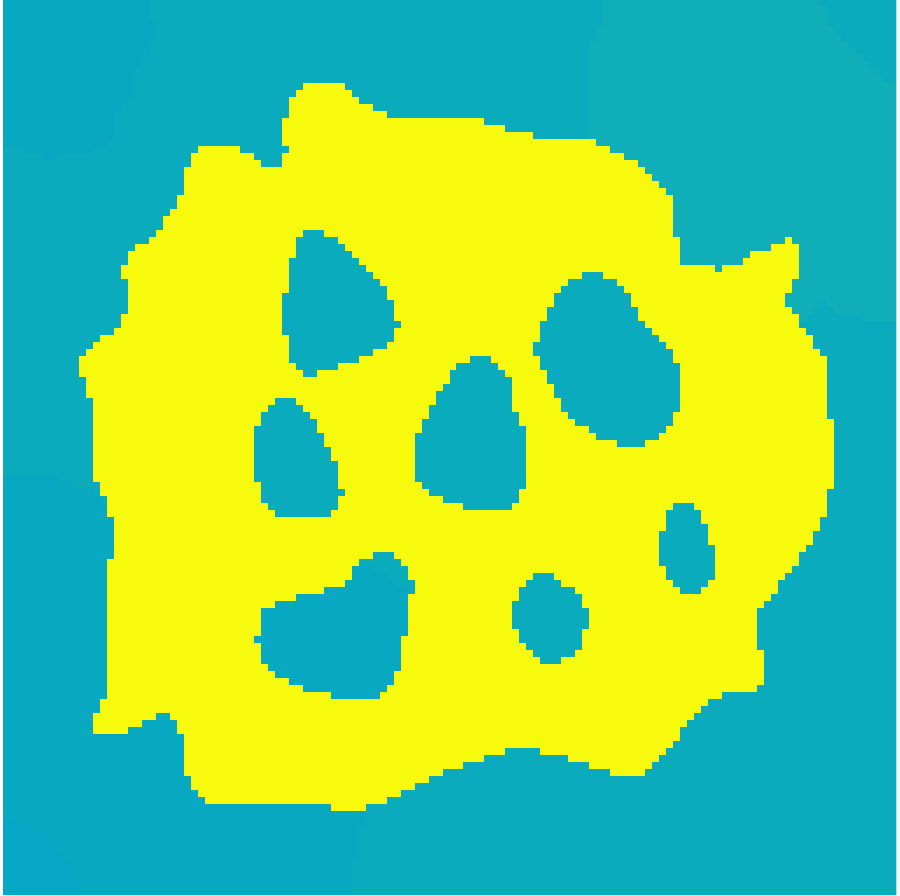} \\
 ($\lambda = 2.97 \times 10^7$) & & ($ \lambda =1.13 \times 10^9$) & & ($\lambda = 2.97 \times 10^7$) & & ($\lambda =1.27 \times 10^8$)
\end{tabular}
\caption{Benchmark Tests: Reconstructions with full-view noiseless data for the regularization parameter $\lambda$ shown below it. }
\label{fig:benchmark}
\end{figure}

\subsection{Limited-angle test}
In this case, we use \textit{only} 5 projections with $\theta$ restricted from $0$ to $2\pi/3$. The data is now reduced to almost $3\%$ compared to the benchmark test. We also add Gaussian noise of 10~dB SNR to this synthetic data.
To check the performance of the proposed method, we compare it to Total-variation method \cite{bleichrodt2016easy}, DART \cite{batenburg2011dart} and its modified version for partially discrete tomography, P-DART \cite{roelandts2012accurate}. A total of 200 iterations were performed with regularization parameter determined from shape residual curve. In DART, the background part was modeled using 20 discrete gray-values between its bounds for model A and B, while 30 discrete gray-values for model C and D. 40 DART iterations were perfomed in each case. For P-DART, a total of 150 iterations were performed.

\begin{figure}
\centering
\renewcommand{\arraystretch}{1.5}
\begin{tabular}{ccccc} 
\textbf{Model} & \textbf{Total-Variation} & \textbf{DART} & \textbf{P-DART} & \textbf{Proposed Method} \\
& ($\lambda = 3.36$) & & &($\lambda = 3.793 \times 10^5$) \\
\includegraphics[width=0.15\columnwidth]{model1.pdf} &\includegraphics[width=0.15\columnwidth]{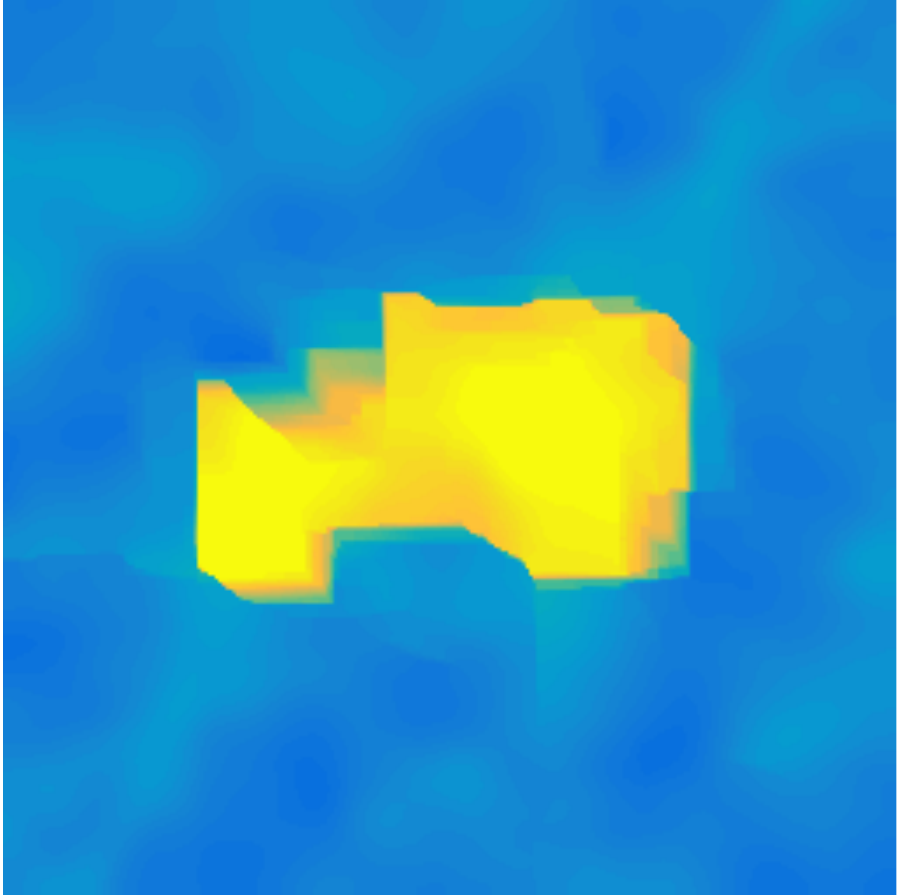} & \includegraphics[width=0.15\columnwidth]{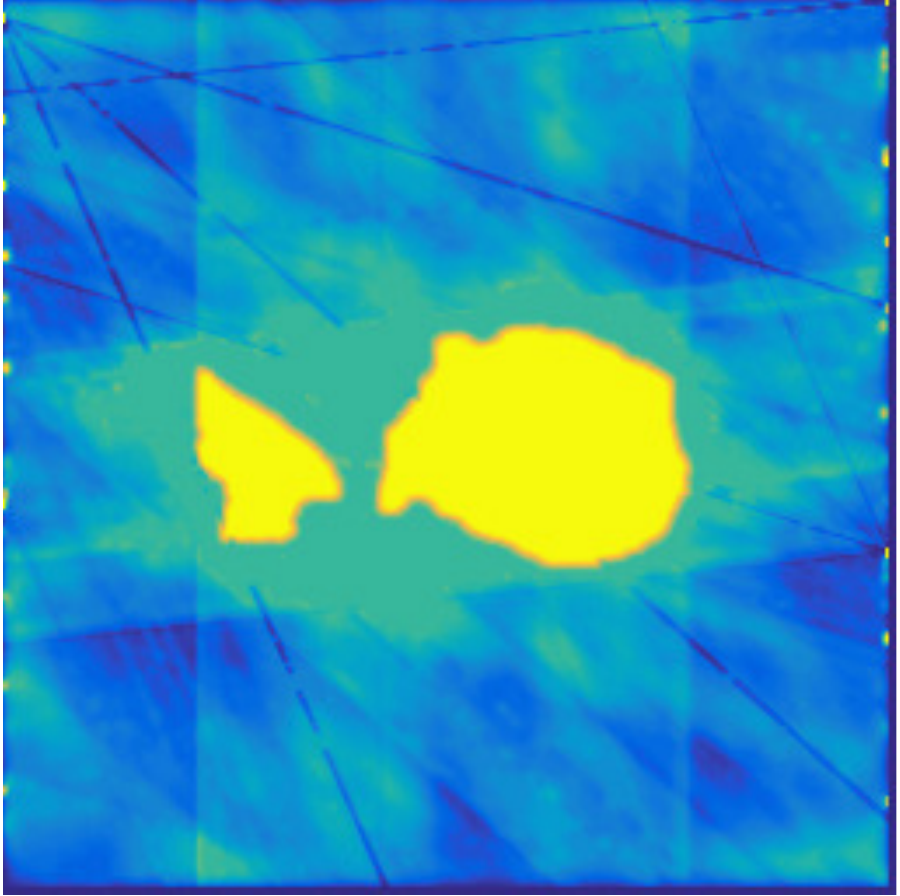} & \includegraphics[width=0.15\columnwidth]{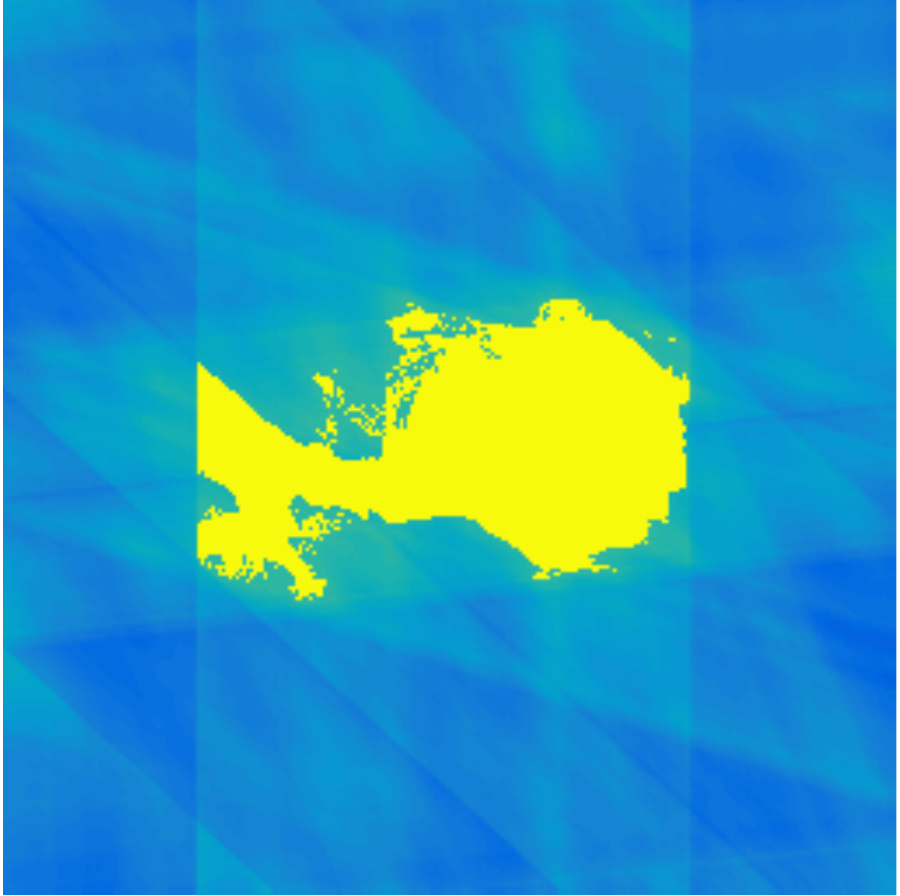} & \includegraphics[width=0.15\columnwidth]{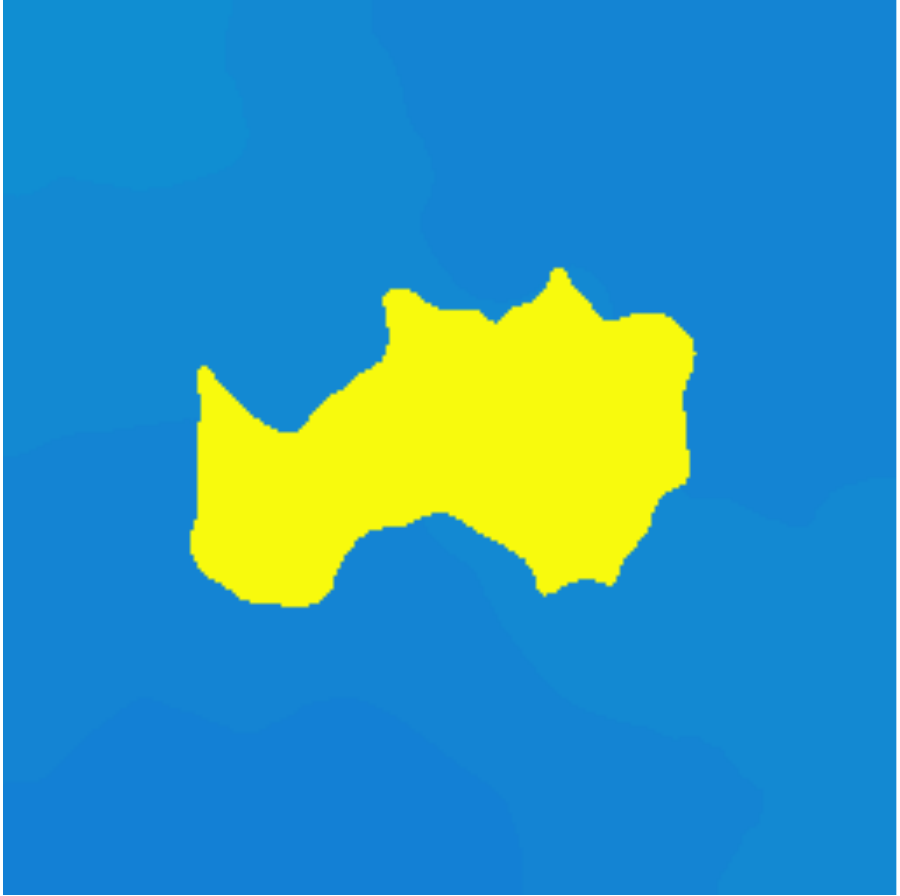} \\
& DR = 59.91 & DR = 119.36 & \textbf{DR = 13.64} & DR = 74.07 \\
& SR = 6904 & SR = 4542 & SR = 2207 & \textbf{SR = 352} \\
& ($\lambda = 1.438$) & & & ($\lambda = 3.793 \times 10^5$) \\
\includegraphics[width=0.15\columnwidth]{model2.pdf}  & \includegraphics[width=0.15\columnwidth]{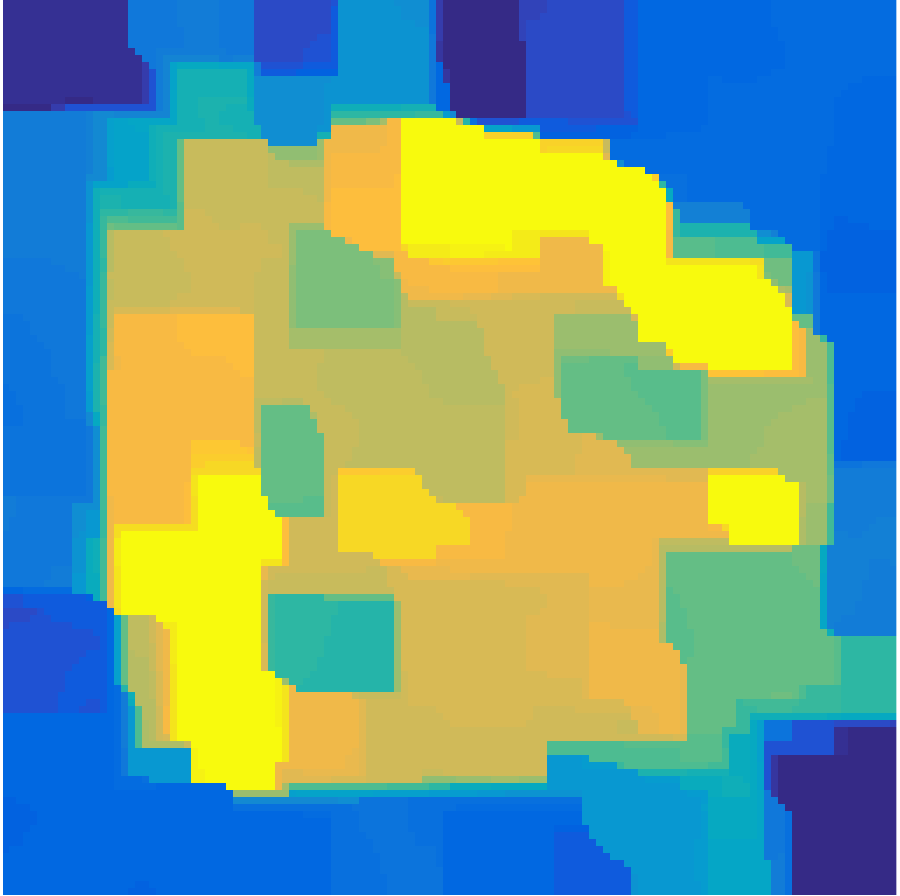} & \includegraphics[width=0.15\columnwidth]{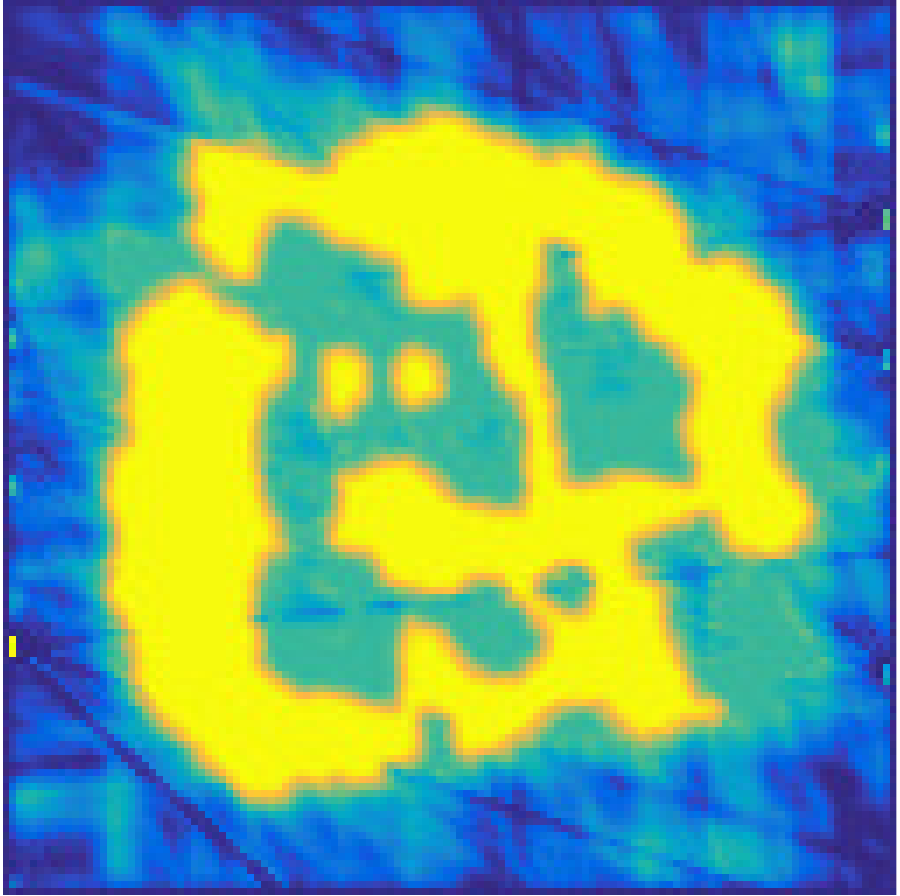} & \includegraphics[width=0.15\columnwidth]{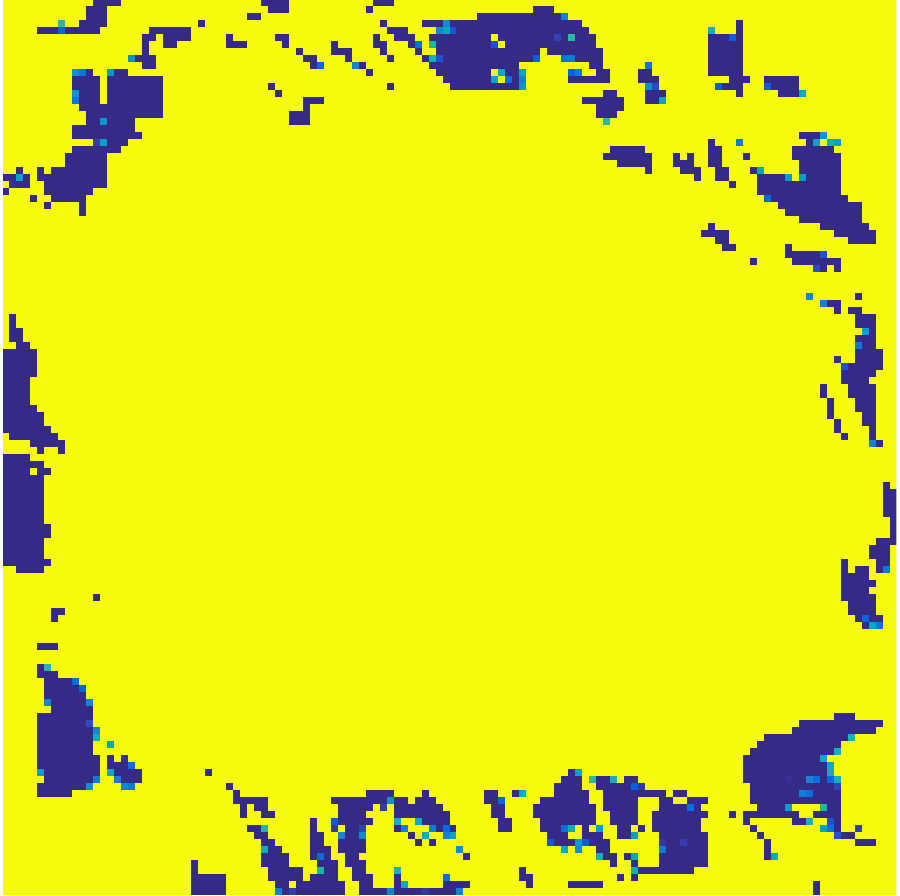} & \includegraphics[width=0.15\columnwidth]{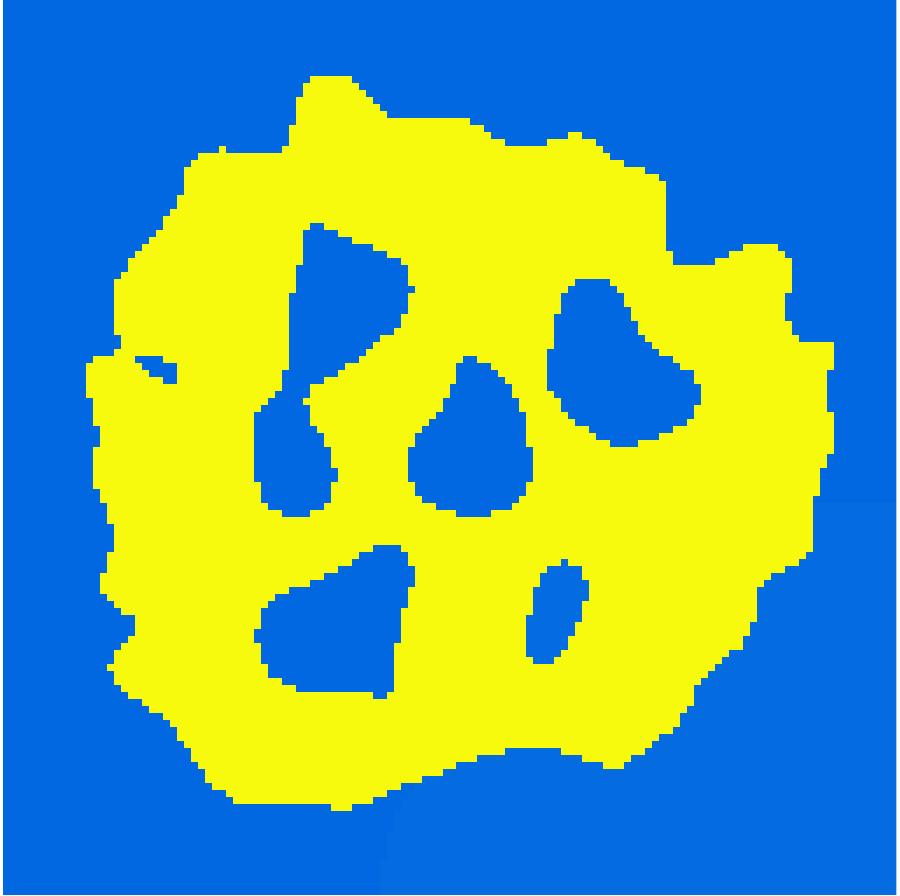} \\
& DR = 39.39 & \textbf{DR = 35.25} & DR = 298.61 & DR = 62.28 \\
& SR = 5366 & SR = 4057 & SR = 7622 & \textbf{SR = 796} \\
& ($\lambda = 3.36$) & & & ($\lambda = 7.438 \times 10^5$) \\
\includegraphics[width=0.15\columnwidth]{model3.pdf} &  \includegraphics[width=0.15\columnwidth]{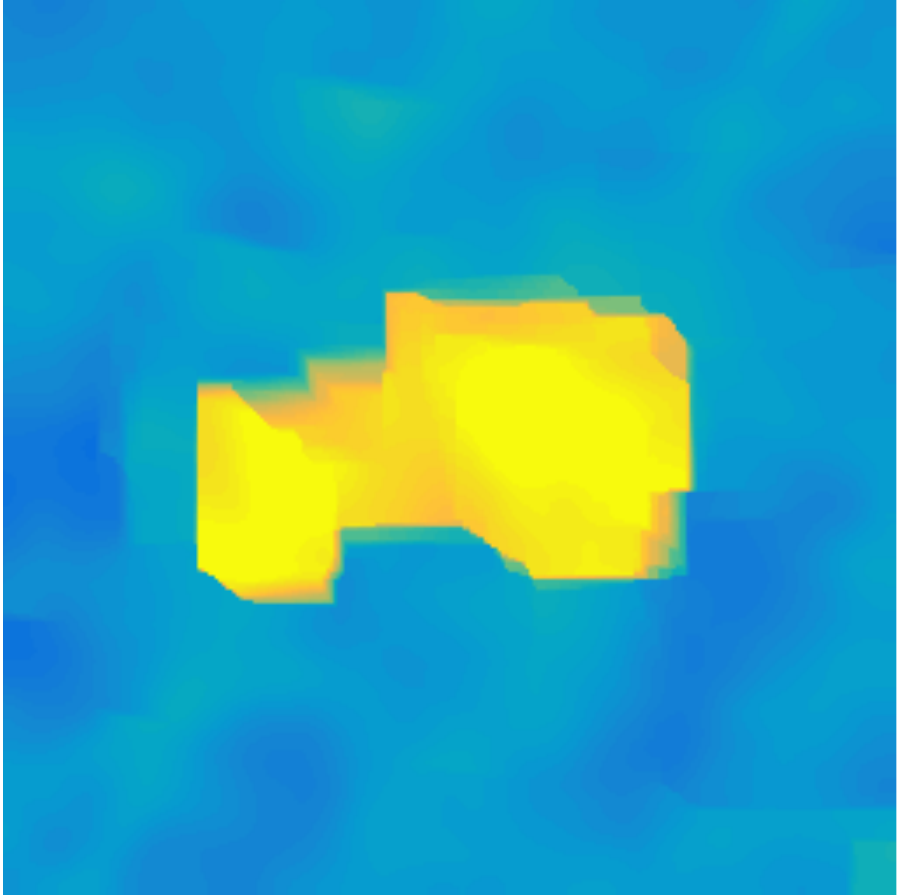} & \includegraphics[width=0.15\columnwidth]{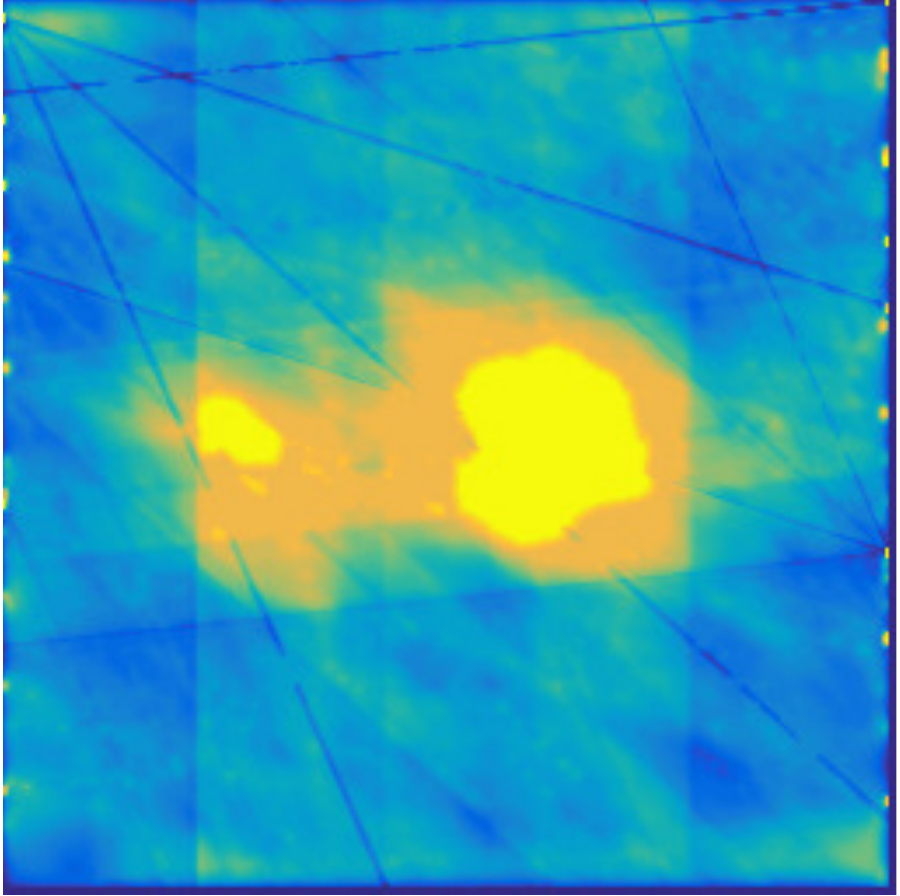} & \includegraphics[width=0.15\columnwidth]{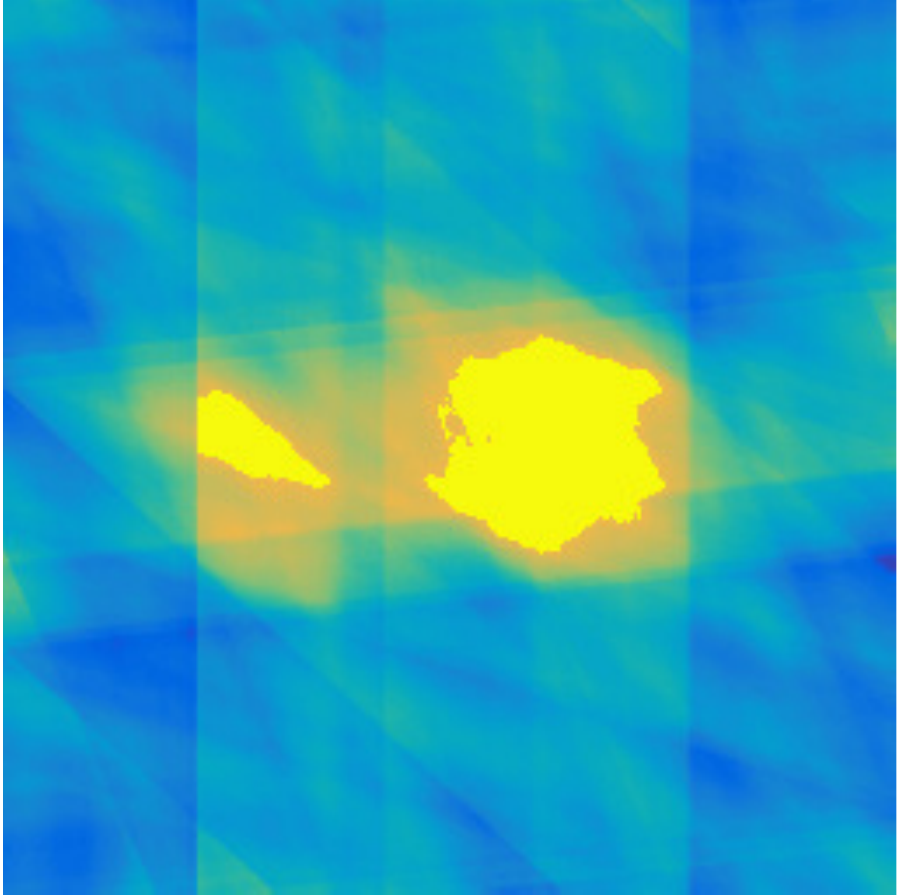} & \includegraphics[width=0.15\columnwidth]{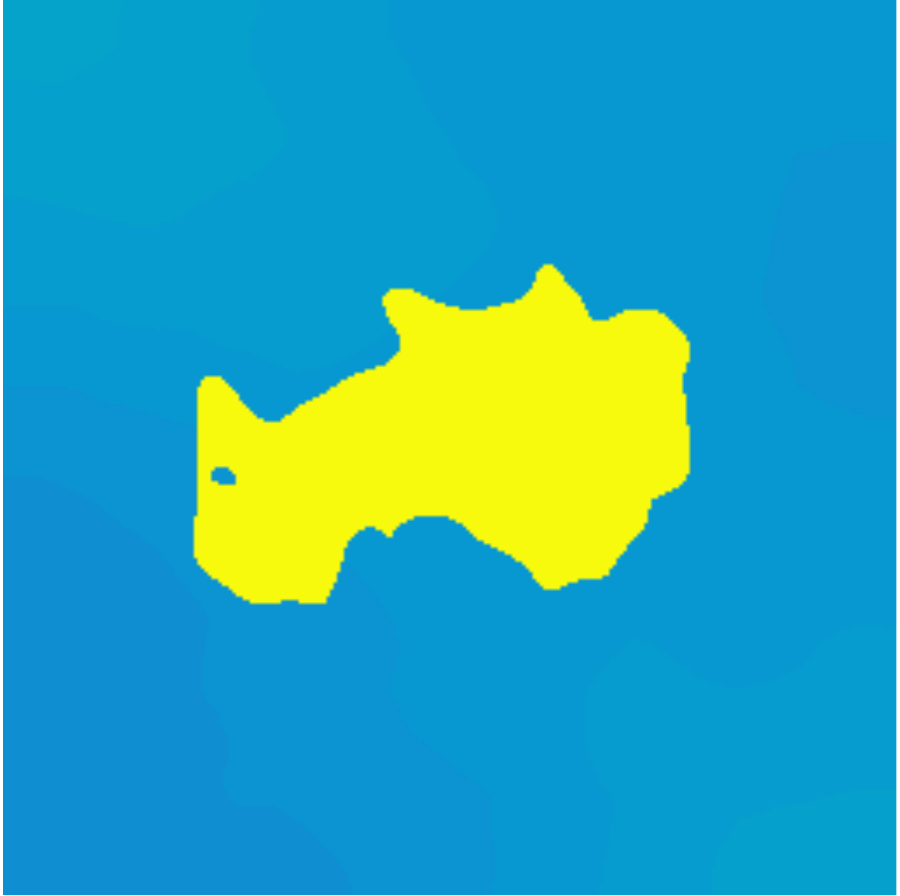} \\
& DR = 70.35 & DR = 134.99 & \textbf{DR = 16.54} & DR = 118.56 \\
& SR = 6805 & SR = 6964 & SR = 5541 & \textbf{SR = 377} \\
& ($\lambda = 0.6158$) & & & ($\lambda = 3.793 \times 10^5$) \\
\includegraphics[width=0.15\columnwidth]{model4.pdf} & \includegraphics[width=0.15\columnwidth]{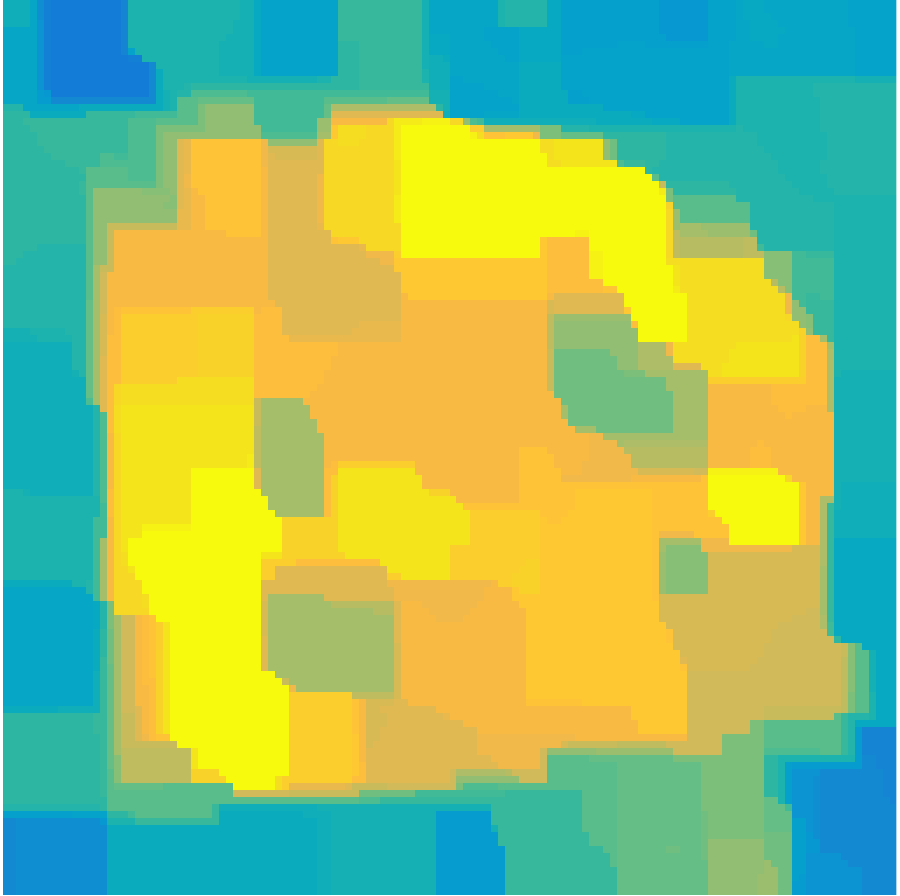} & \includegraphics[width=0.15\columnwidth]{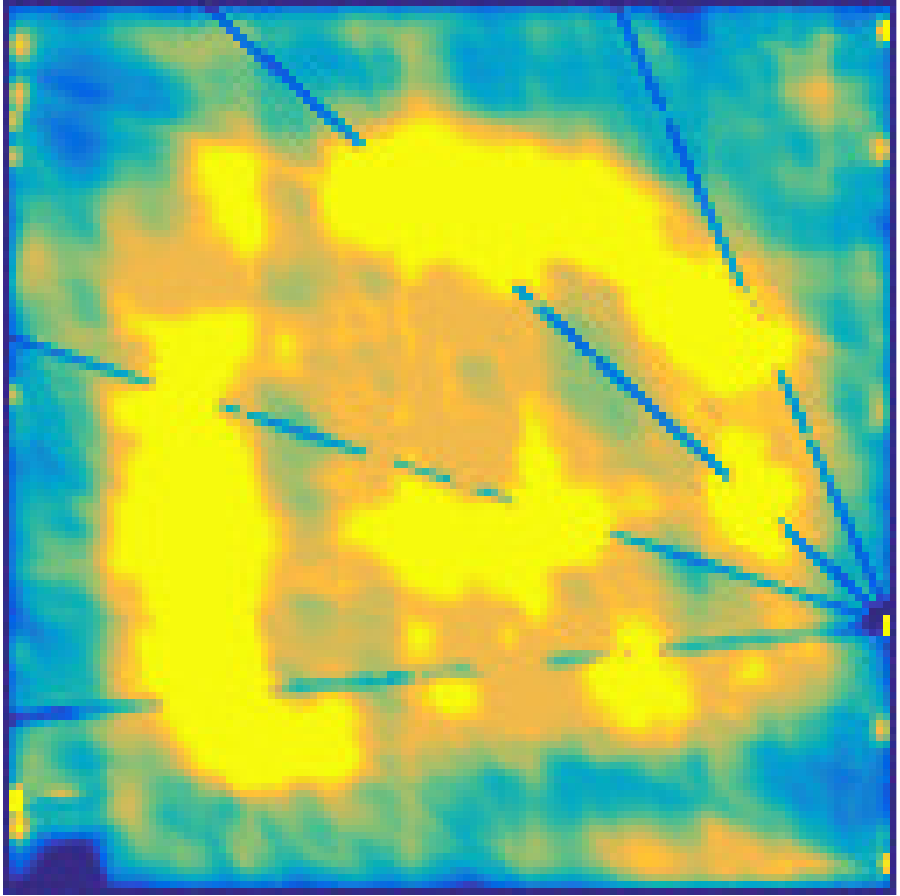} & \includegraphics[width=0.15\columnwidth]{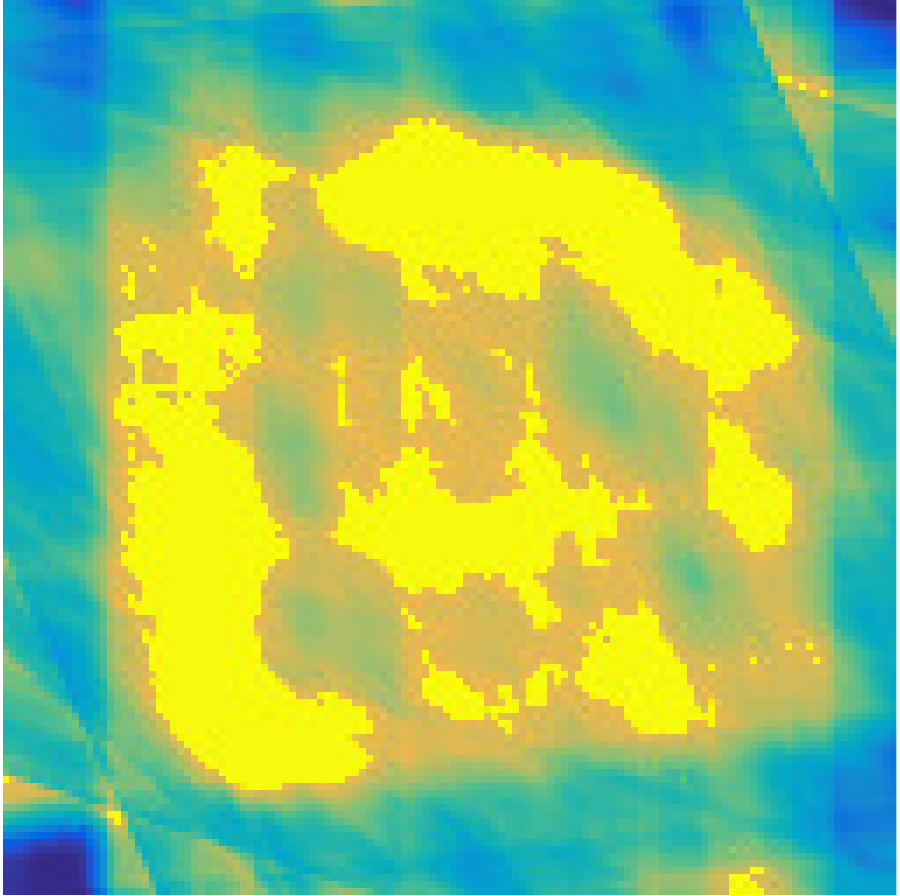} & \includegraphics[width=0.15\columnwidth]{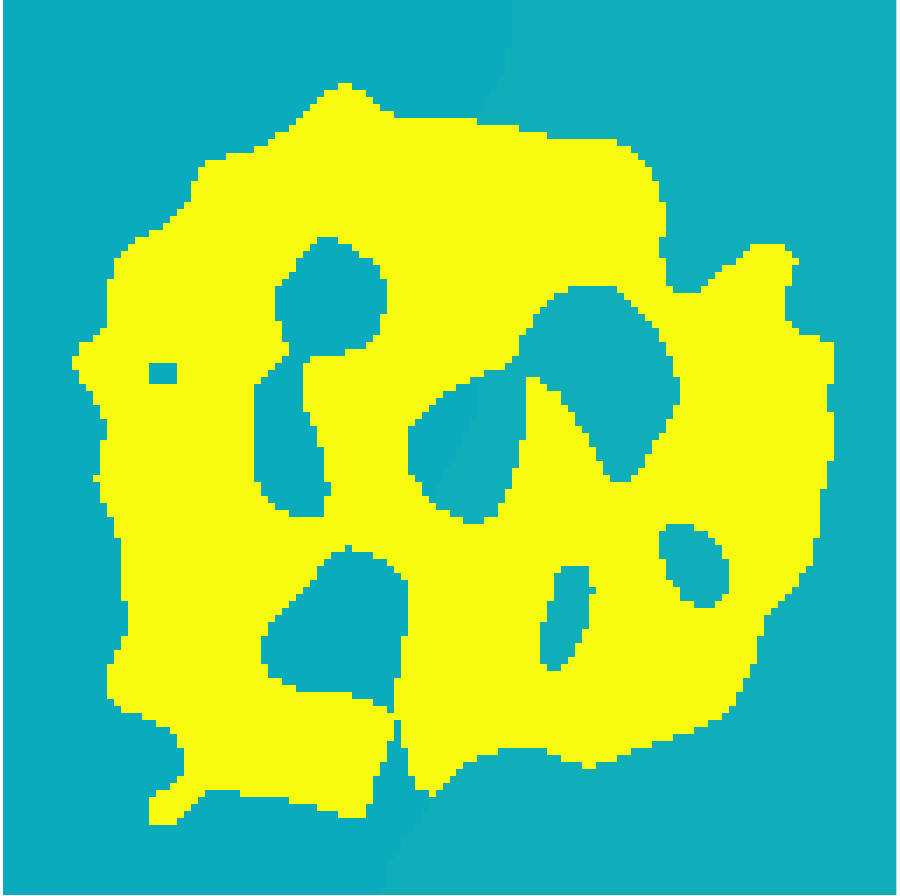} \\
& DR = 21.52 & DR = 96.5216 & \textbf{DR = 9.00} & DR = 59.51 \\
& SR = 5602 & SR = 4766 & SR = 3631 & \textbf{SR = 1304}
\end{tabular}
\caption{Reconstructions with noisy limited data. The first column shows the true models, while the last 4 columns show the reconstructions with various methods. The residuals are also shown below each reconstructed model.}
\label{fig:limited:noise}
\end{figure}

The results on noisy limited-angle with limited data are presented in Figure~\ref{fig:limited:noise}. The proposed method is able to capture most of the fine details (evident from the shape residual) in the phantoms even with the very limited data with moderate noise. The P-DART method achieves the least amount of data residual in all the cases, but fails to capture the complete geometry of the anomaly.

\section{Conclusions and Discussion}
\label{section:discussion}
We discussed a parametric level-set method for partially discrete tomography. We model such objects as a constant-valued shape embedded in a continuously varying background. The shape is represented using a level-set function, which in turn is represented using radial basis functions.
The reconstruction problem is posed as a bi-level optimization problem for the background and level-set parameters. This reconstruction problem can be efficiently solved using a variable projection approach, where the shape is iteratively updated. Each iteration requires a full reconstruction of the background. The algorithm includes some practical heuristics for choosing various parameters that are introduced as part of the parametric level-set method. Numerical experiments on a few numerical phantoms show that the proposed approach can outperform other popular methods for (partially) discrete tomography in terms of reconstruction error. As the proposed algorithm requires repeated full reconstructions, future research is directed at making the method more efficient. 

\textbf{Acknowledgments.} This work is part of the Industrial Partnership Programme (IPP) ‘Computational sciences for energy research’ of the Foundation for Fundamental Research on Matter (FOM), which is part of the Netherlands Organisation for Scientific Research (NWO). This research programme is co-financed by Shell Global Solutions International B.V. The second and third authors are financially supported by the Netherlands Organisation for Scientific Research (NWO) as part of research programmes 613.009.032 and 639.073.506 respectively.
%
%
\bibliographystyle{abbrv}
\bibliography{main.bib}

\begin{thebibliography}{10}

\bibitem{aghasi2011parametric}
A.~Aghasi, M.~Kilmer, and E.~L. Miller.
\newblock Parametric level set methods for inverse problems.
\newblock {\em SIAM Journal on Imaging Sciences}, 4(2):618--650, 2011.

\bibitem{aravkin2012estimating}
A.~Y. Aravkin and T.~Van~Leeuwen.
\newblock Estimating nuisance parameters in inverse problems.
\newblock {\em Inverse Problems}, 28(11):115016, 2012.

\bibitem{batenburg2011dart}
K.~J. Batenburg and J.~Sijbers.
\newblock Dart: a practical reconstruction algorithm for discrete tomography.
\newblock {\em IEEE Transactions on Image Processing}, 20(9):2542--2553, 2011.

\bibitem{bleichrodt2016easy}
F.~Bleichrodt, T.~van Leeuwen, W.~J. Palenstijn, W.~van Aarle, J.~Sijbers, and
  K.~J. Batenburg.
\newblock Easy implementation of advanced tomography algorithms using the astra
  toolbox with spot operators.
\newblock {\em Numerical algorithms}, 71(3):673--697, 2016.

\bibitem{burger2001level}
M.~Burger.
\newblock A level set method for inverse problems.
\newblock {\em Inverse problems}, 17(5):1327, 2001.

\bibitem{chambolle2011first}
A.~Chambolle and T.~Pock.
\newblock A first-order primal-dual algorithm for convex problems with
  applications to imaging.
\newblock {\em Journal of Mathematical Imaging and Vision}, 40(1):120--145,
  2011.

\bibitem{dorn2006level}
O.~Dorn and D.~Lesselier.
\newblock Level set methods for inverse scattering.
\newblock {\em Inverse Problems}, 22(4):R67, 2006.

\bibitem{kadu2016salt}
A.~Kadu, T.~Van~Leeuwen, and W.~A. Mulder.
\newblock Salt reconstruction in full waveform inversion with a parametric
  level-set method.
\newblock {\em IEEE Transactions on Computational Imaging}, 2016.

\bibitem{kak2001principles}
A.~C. Kak and M.~Slaney.
\newblock {\em Principles of computerized tomographic imaging}.
\newblock SIAM, 2001.

\bibitem{klann2011mumford}
E.~Klann, R.~Ramlau, and W.~Ring.
\newblock A mumford-shah level-set approach for the inversion and segmentation
  of spect/ct data.
\newblock {\em Inverse Probl. Imaging}, 5(1):137--166, 2011.

\bibitem{osher2006level}
S.~Osher and R.~Fedkiw.
\newblock {\em Level set methods and dynamic implicit surfaces}, volume 153.
\newblock Springer Science \& Business Media, 2006.

\bibitem{osher1988fronts}
S.~Osher and J.~A. Sethian.
\newblock Fronts propagating with curvature-dependent speed: algorithms based
  on hamilton-jacobi formulations.
\newblock {\em Journal of computational physics}, 79(1):12--49, 1988.

\bibitem{roelandts2012accurate}
T.~Roelandts, K.~Batenburg, E.~Biermans, C.~K{\"u}bel, S.~Bals, and J.~Sijbers.
\newblock Accurate segmentation of dense nanoparticles by partially discrete
  electron tomography.
\newblock {\em Ultramicroscopy}, 114:96--105, 2012.

\bibitem{sidky2008image}
E.~Y. Sidky and X.~Pan.
\newblock Image reconstruction in circular cone-beam computed tomography by
  constrained, total-variation minimization.
\newblock {\em Physics in medicine and biology}, 53(17):4777, 2008.

\bibitem{thompson1991study}
A.~M. Thompson, J.~C. Brown, J.~W. Kay, and D.~M. Titterington.
\newblock A study of methods of choosing the smoothing parameter in image
  restoration by regularization.
\newblock {\em IEEE Transactions on Pattern Analysis and Machine Intelligence},
  13(4):326--339, 1991.

\bibitem{wright1999numerical}
S.~Wright and J.~Nocedal.
\newblock Numerical optimization.
\newblock {\em Springer Science}, 35:67--68, 1999.

\end{thebibliography}

\end{document}